\documentclass[nocrop]{ipart_v1_modif}
\pdfoutput=1
\newcommand\thepage{\arabic{page}}
\usepackage[ascii]{inputenc}
\usepackage[T1]{fontenc}
\usepackage[USenglish]{babel}
\usepackage[bookmarksnumbered=true,pdfpagelabels=true,hypertexnames=false]{hyperref}
\usepackage{microtype,amsthm,cleveref,braket,bbm,booktabs}

\newtheorem{theorem}{Theorem}[section]\crefname{theorem}{theorem}{theorems}
\newtheorem{lemma}[theorem]{Lemma}\crefname{lemma}{lemma}{lemmas}
\newtheorem{proposition}[theorem]{Proposition}\crefname{proposition}{proposition}{propositions}
\newtheorem{corollary}[theorem]{Corollary}\crefname{corollary}{corollary}{corollaries}
\newtheorem{definition}[theorem]{Definition}\crefname{definition}{definition}{definitions}
\theoremstyle{definition}
\newtheorem{example}[theorem]{Example}\crefname{example}{example}{examples}
\numberwithin{equation}{section}
\allowdisplaybreaks[4]

\newtheorem*{maintheorem}{\Cref{thm:main}}
\newtheorem*{hornproposition}{\Cref{prp:horn condition}}
\newcommand{\hornstatement}{For any Ressayre element $H$,
  \[ \kappa_H :=
  \tr \pi(-)\big|_{\calH(H<0)} - \tr \ad(-)\big|_{\mathfrak n_-(H<0)} =
  \!\!\!\!\sum_{\omega \in \Omega : (\omega,H) < 0} \!\!\!\!\omega - \!\!\!\!\sum_{\alpha \in R_{G,-} : (\alpha,H) < 0} \!\!\!\!\alpha \]
  is an element of the moment cone $C_{K(H=0)}(\calH(H=0))$.
  In fact, $\delta_H$ is a lowest weight vector of weight $-\kappa_H$.}

\def\id{\mathbbm{1}}
\newcommand{\CC}{{\mathbb C}}
\newcommand{\PP}{{\mathbb P}}
\newcommand{\RR}{{\mathbb R}}
\newcommand{\ZZ}{{\mathbb Z}}
\newcommand{\calO}{{\mathcal O}}
\newcommand{\calH}{{M}}   
\newcommand{\calK}{M_0}   
\newcommand{\gl}{\mathfrak{gl}}

\newcommand{\norm}[1]{\lVert#1\rVert}
\newcommand{\abs}[1]{\lvert#1\rvert}
\newcommand{\sympl}{\omega_\calH}
\renewcommand{\sl}{{\mathfrak{sl}}}
\renewcommand{\Re}{\operatorname{Re}}
\renewcommand{\Im}{\operatorname{Im}}
\DeclareMathOperator{\Gr}{Gr}
\DeclareMathOperator{\spec}{spec}
\DeclareMathOperator{\cone}{cone}
\DeclareMathOperator{\GL}{GL}
\DeclareMathOperator{\SL}{SL}
\DeclareMathOperator{\U}{U}
\DeclareMathOperator{\SU}{SU}
\DeclareMathOperator{\Sym}{Sym}
\DeclareMathOperator{\tr}{Tr}
\DeclareMathOperator{\Span}{span}
\DeclareMathOperator{\Ad}{Ad}
\DeclareMathOperator{\ad}{ad}

\newcommand{\refappendix}{\hyperref[appendix]{appendix}}

\begin{document}

\author{Mich\`ele Vergne and Michael Walter}
\title[Moment Cones of Finite-Dimensional Representations]{Inequalities for Moment Cones of Finite-Dimensional Representations}
\begin{abstract}
  We give a general description of the moment cone associated with an arbitrary finite-dimensional unitary representation of a compact, connected Lie group in terms of finitely many linear inequalities.
  Our method is based on combining differential-geometric arguments with a variant of Ressayre's notion of a dominant pair.
  As applications, we obtain
  generalizations of Horn's inequalities to arbitrary representations,
  new inequalities for the one-body quantum marginal problem in physics,
  which concerns the asymptotic support of the Kronecker coefficients of the symmetric group,
  and a geometric interpretation of the Howe-Lee-Tan-Willenbring invariants for the tensor product algebra.
\end{abstract}
\maketitle

\section{Introduction}

The study of the convexity properties of the moment map and of its image has a long history in mathematics, starting from Schur and Horn's observation that the diagonal entries of a $d \times d$ Hermitian matrix are always contained in the convex hull of the spectrum \cite{Schur23,Horn54}; cf.~\cite{Kostant73}.
More generally, Atiyah and Guillemin-Sternberg have shown that, for any torus action on a compact, connected Hamiltonian manifold, the image of the moment map is a convex polytope, called the \emph{moment polytope} \cite{Atiyah82, GuilleminSternberg82}. It can be explicitly computed as the convex hull of the images of torus fixed points.
For non-abelian groups, the image of the moment map is no longer convex.
Instead, Kirwan's celebrated convexity theorem asserts that, for the action of an arbitrary compact, connected Lie group on a compact, connected Hamiltonian manifold, the intersection of the image of the moment map with a positive Weyl chamber is a convex polytope \cite{Kirwan84a}.
This is the correct generalization of the moment polytope to non-abelian group actions.
Mumford has given a different proof in the case of projective subvarieties, which relies on a concrete description in terms of the decomposition of the homogeneous coordinate ring into irreducible representations \cite{NessMumford84}; cf.~\cite{GuilleminSternberg82a,Brion87}.
However, no effective general methods are known for the computation of these polytopes (in contrast to the case of torus actions).

In this article, we are concerned with the moment polytope of the complex projective space $\PP(\calH)$ associated with a unitary representation $\Pi$ of a compact, connected Lie group $K$ on a finite-dimensional Hilbert space $\calH$.
Equivalently, we shall study the moment cone of the latter, which is the convex cone spanned by the moment polytope (see \cref{sec:notation} below).
Even in this geometrically straightforward situation, computing the moment cone can be remarkably challenging.
For example, the classical problem of characterizing the eigenvalues $\lambda, \mu, \nu$ of triples of $d \times d$ Hermitian matrices that add up to zero, $A + B + C = 0$, is equivalent to determining the moment cone associated with the representation of $K = \SU(d) \times \SU(d) \times \SU(d)$ on $\calH = \mathfrak{gl}(d) \oplus \mathfrak{gl}(d)$ given by $(g,h,k) \cdot (a, b) = (g a k^{-1}, h b k^{-1})$, also known as the \emph{Horn cone}.
While this observation is mathematically straightforward, a concrete description of the Horn cone in terms of finitely many linear inequalities has only been achieved in \cite{Klyachko98} and was an important step towards the proof of Horn's conjecture \cite{Horn62,KnutsonTao99}.
Similarly, the \emph{one-body quantum marginal problem} in quantum physics amounts to the determination of the moment cone for the action of $K = \SU(a) \times \SU(b) \times \SU(c)$ on $\calH = \CC^a \otimes \CC^b \otimes \CC^c$ by tensor products \cite{Klyachko04,DaftuarHayden04}.
The underlying difficulty in both problems is in the representation theory rather than the geometry:
For example, it is well-known that the Horn cone describes the asymptotic support of the Littlewood-Richardson coefficients \cite{Lidskii82}, while the solution to the one-body quantum marginal problem is given by the asymptotic support of the Kronecker coefficients of the symmetric group \cite{ChristandlMitchison06,Klyachko04,ChristandlHarrowMitchison07,ChristandlDoranKousidisWalter2014}.
Just as for the representation-theoretic coefficients \cite{Murnaghan55,Littlewood58}, the former problem can be understood as special case of the latter \cite{Walter14,ChristandlSahinogluWalter12}.
Both problems can also be phrased in terms of projections of coadjoint orbits \cite{Heckman82,BerensteinSjamaar00,Ressayre10}.
We remark that, locally, moment cones of arbitrary Hamiltonian $K$-manifolds can be described in terms of unitary representations \cite{Marle85,GSnormalform84} and therefore fall into the scenario discussed in this paper.

\medskip

Our main contribution is a clean algebraic description of the moment cone in terms of finitely many linear inequalities.
To state the result, let $\mathfrak t$ denote the Lie algebra of a maximal torus of $T$, $i \mathfrak t^*_+$ a positive Weyl chamber, and $\pi$ the (complex) Lie algebra representation induced by $\Pi$ (see \cref{sec:notation} below for precise definitions). We shall say that $H$ is a \emph{Ressayre element} if
1) the hyperplane $(-,H) = 0$ is spanned by weights of $\calH$ and
2) there exists a vector $\psi \in \calH$ annihilated by $\pi(H)$ such that the ``tangent map'' at $\psi$,
\begin{equation}
\label{eq:tangent map intro}
  \mathfrak n_-(H < 0) \rightarrow \calH(H < 0), \quad X \mapsto \pi(X) \psi,
\end{equation}
is an isomorphism; here, $\calH(H < 0)$ denotes the direct sum of all negative eigenspaces of $\pi(H)$ and $\mathfrak n_-(H < 0)$ the sum of all root spaces for negative roots $\alpha$ such that $(\alpha, H) < 0$ (\cref{dfn:ressayre}).
Note that there are only finitely many Ressayre elements for any given representation $\pi$.
In \cref{sec:facets}, we will prove the following result:

\begin{theorem}
\label{thm:main}
  The moment cone for the $K$-action on $\calH$ is given by
  \begin{equation*}
    C_K = \{ \lambda \in i \mathfrak t^*_+ : (H, \lambda) \geq 0 \text{ for all Ressayre elements $H$} \}.
  \end{equation*}
\end{theorem}

To prove \cref{thm:main}, we show that any facet corresponds to a Ressayre element by studying the moment map, which is quadratic, locally up to second order.
To show that, conversely, any Ressayre element determines a valid inequality, we use Mumford's description of the moment cone as in \cite{Ressayre10}.
Indeed, our notion of a Ressayre element is closely related to Ressayre's notion of a dominant pair.
We note that the description in \cref{thm:main} will typically contain redundancies.
Thus our result differs from \cite{Brion99}, where the non-trivial or ``general'' faces of the moment polytope are characterized precisely at the cost of requiring a recursive strategy for their computation.

If $H$ is a Ressayre element then the domain and codomain of the tangent map \eqref{eq:tangent map intro} necessarily have the same dimension, i.e., $\dim \mathfrak n_-(H < 0) = \dim \calH(H < 0)$. We call this the \emph{trace condition}.
Moreover, note that the determinant $\delta_H$ of \eqref{eq:tangent map intro} (with respect to any fixed pair of bases) is a non-zero polynomial in $\psi \in \calH(H=0)$, the zero eigenspace of $\pi(H)$.
In fact, $\delta_H$ is a canonical (up to scalar multiplication) lowest weight vector for the action of $K(H=0)$, the centralizer of the torus generated by $H$, on the space of polynomials on $\calH(H=0)$.
This implies the following result in \cref{sec:trace and horn condition}, which we call the \emph{Horn condition}:

\begin{proposition}[Horn condition]
\label{prp:horn condition}
\hornstatement
\end{proposition}

Here, $\Omega$ denotes the set of weights of $\calH$ and $R_{G,-}$ denotes the set of negative roots of $G$.
By applying \cref{thm:main} to the lower-dimensional scenario, the Horn condition can be explicitly stated as a set of linear inequalities that have to be satisfied by $\kappa_H$.

Tangent maps and their determinants have been studied in great generality by Ressayre and Belkale from an algebro-geometric point of view \cite{Ressayre10,Ressayre10a,Belkale10}, and our \cref{thm:main,prp:horn condition} can also be deduced from their results.
In these works, the non-vanishing of the determinant has been in turn been translated into a cohomological condition.
In contrast, we propose that, for the purposes of computing moment cones explicitly, it can be useful to instead test the non-vanishing of the determinant directly---either symbolically, which is easily possible in small dimensions, or numerically by using fast algorithms for polynomial identity testing, as we discuss in \cref{sec:computation} below.
The challenge imposed by higher dimensions is rather in finding additional \emph{a priori} constraints on the facets of the moment cone.

We apply our approach to the two paradigmatic examples mentioned above.
In \cref{sec:qmp}, we use it to give a new solution to the one-body quantum marginal problem (equivalently, a description of the asymptotic support of the Kronecker coefficients of the symmetric group for triples of Young diagrams with bounded numbers of rows).
In particular, our method allows us to compute the moment polytope for $\CC^4 \otimes \CC^4 \otimes \CC^4$, which was out of reach with previous methods.
In \cref{sec:horn}, we revisit the classical Horn inequalities from the perspective of our work.
We find that they are instances of the trace and Horn conditions derived above, which justifies our terminology.
In other words, our trace condition and Horn condition can be understood as generalizations of Horn's inequalities to arbitrary representations.
We also give a geometric explanation of the invariants constructed in \cite{HoweTanWillenbring2005,HoweLee2007}: they can be obtained directly from the determinant polynomial $\delta_H$.
In both cases, we find that the description of \cref{thm:main} can be readily refined to the mathematical scenario at hand.
We hope that our method similarly provides a useful tool for the study of moment cones in other mathematical applications.
A preliminary version of this work has appeared in the thesis of MW \cite{Walter14}.


\section{Moment Cones of Finite-Dimensional Representations}
\label{sec:notation}
Let $K$ be a compact, connected Lie group with Lie algebra $\mathfrak k$.
Let $G$ be its complexification, which is a connected reductive algebraic group $G$ with Lie algebra $\mathfrak g = \mathfrak k \oplus i \mathfrak k$ and denote its exponential map by $\exp \colon \mathfrak g \rightarrow G$.
Let $T \subseteq K$ be a maximal torus with Lie algebra $\mathfrak t \subseteq \mathfrak k$ and $W_K = N_K(T)/T$ the Weyl group; we write $r_K := \dim T$ for the rank of $K$.
The complexification of $T$ is a maximal abelian subgroup $T_\CC \subseteq G$ with Lie algebra $\mathfrak h = \mathfrak t \oplus i \mathfrak t$.
In view of $\mathfrak u(1) = i \RR$ we consider the weight lattice $P_G$ as a subset of $i \mathfrak t^* \cong \{ \omega \in \mathfrak h^* : \omega(i \mathfrak t) \in \RR \}$.
We denote the set of roots by $R_G \subseteq P_G$ and write the root space decomposition as $\mathfrak g = \mathfrak h \oplus \bigoplus_{\alpha \in R_G} \mathfrak g_{\alpha}$.
For each root $\alpha$, we can find basis vectors $E_\alpha \in \mathfrak g_\alpha$ and ``co-roots'' $H_\alpha \in i \mathfrak t$ that satisfy the commutation relations of $\mathfrak{sl}(2)$,
\begin{equation}
\label{eq:onebody/sl2}
  [E_\alpha, E_{-\alpha}] = H_\alpha
  \quad\text{and}\quad
  [H_\alpha, E_{\pm\alpha}] = \pm 2 E_{\pm\alpha}.
\end{equation}
Let $R_{G,+}$ denote a choice of positive roots. Correspondingly, we get the negative roots $R_{G,-} = -R_{G,+}$, nilpotent Lie algebras $\mathfrak n_\pm = \bigoplus_{\alpha \in R_{G,\pm}} \mathfrak g_\alpha$, maximal unipotent subgroups $N_\pm \subseteq G$ and a positive Weyl chamber, which we take to be the convex polyhedral cone $i \mathfrak t^*_+ = \{ \lambda \in i \mathfrak t^* : (H_\alpha, \lambda) \geq 0 ~ \forall \alpha \in R_{G,+} \} \subseteq i \mathfrak t^*$ with relative interior $i \mathfrak t^*_{>0} = \{ \lambda \in i \mathfrak t^* : (H_\alpha, \lambda) > 0 ~ \forall \alpha \in R_{G,+} \}$.
We may consider $\mathfrak h^* \subseteq \mathfrak g^*$ and $i \mathfrak t^* \subseteq i \mathfrak k^*$ by extending each functional by zero on the root spaces $\mathfrak g_\alpha$ and $\mathfrak k_\alpha = (\mathfrak g_\alpha \oplus \mathfrak g_{-\alpha}) \cap \mathfrak k$, respectively.
The finite-dimensional irreducible representations of $G$ are parametrized by their highest weight $\lambda$ in $P_{G,+} = P_G \cap i \mathfrak t^*_+$; the corresponding representation will be denoted by $V_{G,\lambda}$.
Then $V_{G,\lambda}^* \cong V_{G,\lambda^*}$ for $\lambda^* := -w_0 \lambda$; here and throughout this paper, $w_0$ denotes the longest Weyl group element, which exchanges the positive and negative roots (the group will always be clear from the context).

For $\SU(d)$, whose complexification is $\SL(d)$, we will always use the maximal torus $T(d)$ that consists of the diagonal unitary matrices of unit determinant. Its Lie algebra will be denoted by $\mathfrak t(d)$. We will use as positive roots the $\alpha_{i,j}(H) := H_{ii} - H_{jj}$ with $i < j$, and abbreviate the (positive) roots by $R_d$ and $R_{d,+}$, respectively. Finally, we write $\calO^d_\lambda$ for the coadjoint $\SU(d)$ orbit through a highest weight $\lambda$.

\medskip

Now let $\Pi \colon G \rightarrow \GL(\calH)$ a representation of $G$ on a finite-dimensional Hilbert space $\calH$ that is equipped with a $K$-invariant Hermitian inner product $\braket{-|-}$, which we take to be antilinear in the \emph{first} argument.
We will oftentimes use Dirac's notation $\braket{\phi | A | \psi} := \braket{\phi | A \psi}$ for $\phi, \psi \in \calH$ and $A \in \mathfrak{gl}(\calH)$.
We denote by $\pi \colon \mathfrak g \rightarrow \gl(\calH)$ the induced Lie algebra representation, by $\Omega \subseteq P_G$ the set of weights and write the weight space decomposition as $\calH = \bigoplus_{\omega \in \Omega} \calH_\omega$.
The $K$-action on $\calH$ admits a canonical (up to conventions) \emph{moment map}, defined by
\begin{equation}
\label{eq:moment map}
  \mu_K \colon \calH \rightarrow i \mathfrak k^*, \quad (\mu_K(\psi), X) = \braket{\psi | \pi(X) | \psi}
\end{equation}
for all $\psi \in \calH$ and $X \in i \mathfrak k$. 
Here and in the following, we write $(\varphi, X) = \varphi(X)$ for the duality pairing. 
The map $\mu_K$ is indeed a moment map in the sense of symplectic geometry: it is $K$-invariant and satisfies the basic identity
\begin{equation}
\label{eq:moment map identity}
  d(\mu_K, iX)\big|_\psi = \sympl(\pi(X)\psi, -)
\end{equation}
for all $X \in \mathfrak k$, where $\sympl(\phi, \psi) = 2 \Im \braket{\phi | \psi}$ denotes the symplectic form that we will use for $\calH$.
The \emph{moment cone} then is defined as intersection of the moment map image with the positive Weyl chamber,
\[ C_K = \mu_K(\calH) \cap i \mathfrak t^*_+ = \{ \lambda \in i \mathfrak t^*_+ : \lambda \in \mu_K(\calH) \}. \]

The representation of $G$ also induces an action on the complex projective space $\PP(\calH)$, $g \cdot [\psi] = [g \psi]$, with corresponding moment map
$\widetilde\mu_K \colon \PP(\calH) \rightarrow i \mathfrak k^*, (\widetilde\mu_K([\psi]), X) = {\braket{\psi | \pi(X) | \psi}} / {\braket{\psi | \psi}}$
for all $[\psi] \in \PP(\calH)$ and $X \in i\mathfrak k$ \cite{Kirwan84,NessMumford84}.
Kirwan's convexity theorem \cite{Kirwan84a} (or Mumford's version \cite{NessMumford84}) asserts that the \emph{moment polytope}
$\Delta_K = \widetilde\mu_K(\PP(\calH)) \cap i \mathfrak t^*_+$
is indeed a convex polytope.
It is plain from the definitions that $C_K = \RR_+ \Delta_K$.
Therefore, the moment cone $C_K$ is a polyhedral cone, and we have the following representation-theoretic description \cite{NessMumford84,Brion87},
\begin{equation}
\label{eq:mumford description}
  C_K = \cone \{ \lambda \in P_{G,+} : V_{G,\lambda}^* \subseteq R(\calH) \},
\end{equation}
where $R(\calH) = \Sym(\calH)^*$ denotes the space of polynomials on $\calH$.
A representation $V_{G,\lambda}^*$ occurs in $R(\calH)$ if and only if there exists a polynomial $P \in R(\calH)$ of weight $-\lambda$ that is invariant under the action of the lower unipotent subgroup $N_-$.
We shall call such a polynomial a \emph{lowest weight vector} in $R(\calH)$.

Conversely, suppose that the action of $G$ contains the multiplication by scalars, generated by some $J \in i \mathfrak t$ such that $\pi(J) = \id$ (this can always be arranged for by adding a $\CC^*$-factor to $G$).
Then the moment polytope can be recovered from the moment cone by
\[ \Delta_K = \{ \lambda \in C_K : (\lambda, J) = 1 \} \]
(and $C_K$ is a pointed cone with base $\Delta_K$).
Thus we may equivalently study moment cones of representations or moment polytopes of the corresponding projective spaces.

Throughout this paper, we shall always work with moment cones (but see \cite{Walter14} for an exposition from the projective point of view).
We shall moreover assume that the moment cone is of maximal dimension, i.e., $\dim C_K = \dim i \mathfrak t^*$.
This is the case if and only if there exists a vector with finite stabilizer.

\section{Facets of the Moment Cone}
\label{sec:facets}

Like any polyhedral cone, the moment cone can be described by finitely many linear inequalities $(-,H) \geq 0$.
Since we have assumed that $C_K$ is of maximal dimension, its facets are of codimension one in $i \mathfrak t^*$ and their inward-pointing normal vectors may be identified with the defining linear inequalities $(-,H) \geq 0$ of the moment cone.
Since the moment cone is obtained by intersecting $\mu_K(\calH)$ with the positive Weyl chamber, which itself is a maximal-dimensional polyhedral cone, some of the facets of $C_K$ can be subsets of facets of $i \mathfrak t^*_+$, and we shall call those the trivial facets of the moment cone:

\begin{definition}
  A facet of the moment cone is \emph{trivial} if it corresponds to an inequality of the form $(-, H_\alpha) \geq 0$ for some positive root $\alpha \in R_{G,+}$.
  Otherwise, the facet is called \emph{non-trivial}.
\end{definition}

Non-trivial facets have also been called ``general'' in the literature \cite{Brion99}.
We record the following straightforward observation:

\begin{lemma}
\label{lem:non-trivial facets}
  Any non-trivial facet of $C_K$ meets the relative interior $i \mathfrak t^*_{>0}$ of the positive Weyl chamber.
\end{lemma}

\subsection{Admissibility}
\label{subsec:torus action}

We first consider the moment map $\mu_T$, defined as in \eqref{eq:moment map} for the action of the maximal torus $T \subseteq K$.
Let $\calH = \bigoplus_{\omega \in \Omega} \calH_\omega$ be the decomposition of $\calH$ into weight spaces and let $\psi \in \calH$ be a vector decomposed accordingly as $\psi = \sum_\omega \psi_\omega v_\omega$.
Then $\mu_T$ has the following concrete description:
\begin{equation}
  \label{eq:abelian conic combination}
  \mu_T(\psi) = \sum_\omega \abs{\psi_\omega}^2 \omega
\end{equation}
Observe that $\mu_T(\psi)$ is a conic combination of weights.
It follows that the ``abelian'' moment cone $C_T$ of $\calH$ is precisely equal to the conical hull of the set of weights; it is maximal-dimensional since it contains $C_K$.
More generally, if $\Omega' \subseteq \Omega$ is a subset of weights and $\calH_{\Omega'} := \bigoplus_{\omega \in \Omega'} \calH_\omega$ then $C_T(\calH_{\Omega'}) = \cone \Omega'$.
For the next lemma recall that a critical point of a smooth map $f \colon M \rightarrow M'$ is a point $m \in M$ where the differential $df\big|_m$ is not surjective; a critical value is the image $f(m)$ of a critical point. 
Then the following is well-known (e.g., \cite[Remark 4.14]{ChristandlDoranKousidisWalter2014}):

\begin{lemma}
  \label{lem:abelian critical}
  The set of critical values of $\mu_T$ is equal to the union of the codimension-one conic hulls of subsets of weights.
\end{lemma}
\begin{proof}
  Let $\psi \in \calH$ with weight decomposition $\psi = \sum_\omega \psi_\omega v_\omega$.
  By \eqref{eq:abelian conic combination}, $\mu_T(\psi)$ is a conic combination of weights in $\Omega_\psi := \{ \omega \in \Omega : \psi_\omega \neq 0 \}$.
  By the moment map property \eqref{eq:moment map identity}
  and non-degeneracy of the symplectic form, $\psi$ is a critical point if and only if there exists $0 \neq X \in \mathfrak t$ such that $\pi(X) \psi = 0$ \cite[Lemma 2.1]{GuilleminSternberg82}, i.e., if and only if $\omega(X) = 0$ for all $\omega \in \Omega_\psi$.
  It follows that $\psi$ is a critical point if and only if the conic hull of $\Omega_\psi$ is of positive codimension.

  In particular, any critical value is contained in a codimension-one conic hull of weights, since we may always add additional weights.
  Conversely, if $\Omega' \subseteq \Omega$ is a subset of weights that spans a conic hull of codimension one then $C_T(\calH_{\Omega'}) = \cone \Omega'$ consists of critical values.
\end{proof}

We now derive a basic necessary condition that cuts down the defining inequalities of the moment cone to a finite set of candidates (cf.~\cite[Remark 3.6]{ChristandlDoranKousidisWalter2014}).

\begin{definition}
  An element $H \in i \mathfrak t$ is called \emph{admissible} if the linear hyperplane $(-, H) = 0$ is spanned by a subset of weights in $\Omega$.
\end{definition}

The notion of admissibility is invariant under the action of the Weyl group $W_K$.

\begin{lemma}[Admissibility condition]
\label{lem:admissible}
  Let $(-,H) \geq 0$ be an inequality corresponding to a non-trivial facet of the moment cone.
  Then $H$ is admissible.
\end{lemma}
\begin{proof}
  By \cref{lem:non-trivial facets}, the intersection of $(-,H) = 0$ with the interior of the positive Weyl chamber $i \mathfrak t^*_{>0}$ is non-empty.
  Each point in this intersection is a critical value for $(\mu_K,H) = (\mu_T,H)$, hence of $\mu_T$, and therefore according to \cref{lem:abelian critical} contained in a linear hyperplane spanned by a subset of weights.
  Since this is true for all points in the intersection, which contains the relative interior of the facet, it follows that the facet is in fact contained in a single such hyperplane.
\end{proof}

\subsection{Description of the Moment Cone by Ressayre Elements}

Let us now fix an inequality $(-,H) \geq 0$ corresponding to a non-trivial facet $(-,H) = 0$ of the moment cone.
Let $\psi \in \calH$ be a preimage of a point $\mu_K(\psi) \in i \mathfrak t^*_{>0}$ on the facet $(-,H) = 0$.
In the proof of \cref{lem:admissible}, we have used that $\psi$ is a critical point of $(\mu_K, H)$ (equivalently, that $\pi(H) \psi = 0$) to gain information on the set of possible facets.
To study the function $(\mu_K, H)$ in the vicinity of such a critical point $\psi$ it is natural to consider the Hessian, which is the quadratic form
\begin{equation}
\label{eq:hessian}
  Q(V, V) = 2 \braket{V | \pi(H) | V}.
\end{equation}
For tangent vectors generated by the infinitesimal action of $X, Y \in \mathfrak k$, we have the formula
\begin{equation}
\label{eq:hessian explicit}
\begin{aligned}
  Q(\pi(X)\psi, \pi(Y)\psi)
  &= -2 \Re \braket{\psi | \pi(X) \pi(H) \pi(Y) | \psi} \\
  &= \braket{\psi | \pi([[H, X], Y]) | \psi}
  = (\mu_K(\psi), [[H, X], Y]),
\end{aligned}
\end{equation}
where we have used that $\pi(H) \psi = 0$.
%
%
We now decompose
\begin{equation}
\label{eq:weight space decomposition}
  \calH = \calH(H < 0) \oplus \calH(H = 0) \oplus \calH(H > 0),
\end{equation}
where $\calH(H < 0) = \bigoplus_{\omega : (\omega,H) < 0} \calH_\omega$ is the sum of the eigenspaces of the Hermitian operator $\pi(H)$ with eigenvalue less than $0$, etc.
Then it is plain from \eqref{eq:hessian} that the index of the Hessian $Q$, i.e., the dimension of a maximal subspace on which the quadratic form is negative definite, is equal to twice the complex dimension of $\calH(H<0)$.

\medskip

We will now deduce a second formula for the index by observing that the Hessian is necessarily positive semidefinite on the subspace of those tangent vectors $V$ that are mapped to $i \mathfrak t^*$ by the differential of the moment map.
To see this, 
consider a curve $\psi_t$ with $\psi_0 = \psi$, $\dot\psi_0 = V$ and $\mu_K(\psi_t) \in i \mathfrak t^*_{>0}$ for all $t \in (-\varepsilon,\varepsilon)$ (such a curve can always be constructed by using the symplectic cross section \cite[Theorem 26.7]{GuilleminSternberg84}).
Then, since $(\mu_K(\psi), H) = 0$ and $d(\mu_K,H)\big|_\psi \equiv 0$,
\[ (\mu_K(\psi_t), H) = \frac {t^2} 2 Q(V, V) + O(t^3). \]
But $(\mu_K(\psi_t), H) \geq 0$, since $\mu_K(\psi_t) \in C_K$ and $(-,H) \geq 0$ is a valid inequality for the moment cone.
Together, this shows that, indeed, $Q(V, V) \geq 0$.

The subspace of all such $V$ can be computed in a different way.
For this, let $\mathfrak r = \bigoplus_{\alpha \in R_{G,+}} \mathfrak k_\alpha$ denote the sum of the root spaces of the compact Lie algebra. 
It follows from the moment map property \eqref{eq:moment map identity} that
\[
  d\mu_K\big|_\psi(V) \in i \mathfrak t^*
  \;\Leftrightarrow\;
  0 = d(\mu_K, iR)\big|_\psi(V) = -\sympl(V, \pi(R)\psi) \quad (\forall R \in \mathfrak r).
\]
We thus obtain the following lemma.

\begin{lemma}
\label{lem:cross section hessian positive semidefinite}
  The Hessian is positive semidefinite on the symplectic complement
  \begin{align*}
    \calK^{\sympl}
    &= \{ V \in \calH : \sympl(V, W) = 0 ~ (\forall W \in \calK) \} \\
    &= \{ V \in \calH : d\mu_K\big|_\psi(V) \in i \mathfrak t^* \}
  \end{align*}
  of $\calK := \pi(\mathfrak r) \psi \subseteq \calH$.
\end{lemma}

In fact, it is well-known that $\calK$ is a symplectic subspace of $\calH$ (see, e.g., \cite[Lemma~6.7]{GuilleminSternberg82}).
To see this, note that the restriction of the symplectic form to $\calK$ is given by
\[ \sympl(\pi(X) \psi, \pi(Y) \psi) = (\mu_K(\psi), [X, Y]) \]
and can therefore be identified with the Kirillov-Kostant-Souriau symplectic form on the coadjoint orbit through $\mu_K(\psi) \in i \mathfrak t^*_{>0}$.
As a consequence, we have the decomposition $\calH = \calK \oplus \calK^{\sympl}$.

\begin{lemma}
\label{lem:tangent space decomposition}
The Hessian is block-diagonal with respect to the decomposition $\calH = \calK \oplus \calK^{\sympl}$.
\end{lemma}
\begin{proof}


  For all $\pi(R) \psi \in \calK$ and $V \in \calK^{\sympl}$ we have that
  \begin{align*}
    Q(\pi(R) \psi, V)
    &= -2 \Re \braket{\psi | \pi(R) \pi(H) | V}
    = 2 \Re \braket{\psi | \pi([H,R]) | V} \\
    &= \sympl(\pi([-iH,R]) \psi, V)
    = 0,
  \end{align*}
  since $[-iH,R] \in \mathfrak r$ and therefore $\pi([-iH,R]) \psi \in \calK$.
\end{proof}

\begin{lemma}
\label{lem:torus complement injective}
  The tangent map $\mathfrak r \rightarrow \calH$, $R \mapsto \pi(R) \psi$ is injective.
\end{lemma}
\begin{proof}
  The stabilizer of the coadjoint action of $K$ at any $\lambda \in i\mathfrak t^*_{>0}$ is $T$, while $\mathfrak k = \mathfrak t \oplus \mathfrak r$.
  As $\mu_K(\psi) \in i \mathfrak t^*_{>0}$, the claim follows from this and $K$-equivariance of the moment map.
\end{proof}

\begin{lemma}
  \label{lem:index hessian roots}
  The index of the Hessian is equal to twice the number of positive roots $\alpha \in R_{G,+}$ such that $(\alpha, H) > 0$.
\end{lemma}
\begin{proof}
  Since $Q$ is positive semidefinite on $\calK^{\sympl}$ (\cref{lem:cross section hessian positive semidefinite}) and block-diagonal with respect to the decomposition $\calH = \calK \oplus \calK^{\sympl}$ (\cref{lem:tangent space decomposition}), it suffices to compute the index of $Q$ on $\calK = \{ \pi(R) \psi : R \in \mathfrak r \}$.
  For this, recall from \eqref{eq:hessian explicit} that $Q(\pi(R)\psi, \pi(S)\psi) = (\mu_K(\psi), [[H, R], S])$ for all $R, S \in \mathfrak r$.
  Since the tangent map $R \mapsto \pi(R) \psi$ is injective (\cref{lem:torus complement injective}), we may instead consider the form
  \begin{equation*}
    \widetilde Q(R, S) := (\mu_K(\psi), [[H, R], S])
  \end{equation*}
  on $\mathfrak r$.
  Now observe that $\widetilde Q$ is block-diagonal with respect to $\mathfrak r = \bigoplus_{\alpha \in R_{G,+}} \mathfrak k_\alpha$, since for all $R \in \mathfrak k_\alpha$ and $S \in \mathfrak k_\beta$, $[[H, R], S] \in i\mathfrak k_{\alpha\pm\beta}$, while $\mu_K(\psi) \in i \mathfrak t^*$.
  Therefore, it suffices to compute the index on a single root space $\mathfrak k_\alpha$.
  For this, define the ``Pauli matrices'' $X_\alpha := E_\alpha + E_{-\alpha}$ and $Y_\alpha := i (E_{-\alpha} - E_\alpha)$, which satisfy the commutation relations $[X_\alpha, Y_\alpha] = 2 i H_\alpha$ etc.
  Then $iX_\alpha$ and $iY_\alpha$ form a basis of $\mathfrak k_\alpha$ and
  \begin{align*}
    \widetilde Q(iX_\alpha, iX_\alpha)
    &= -(\mu_K(\psi), [[H, X_\alpha], X_\alpha])
    = -i (\alpha, H) \, (\mu_K(\psi), [Y_\alpha, X_\alpha]) \\
    &= -2 (\alpha, H) \, (\mu_K(\psi), H_\alpha).
  \end{align*}
  Likewise, $\widetilde Q(iY_\alpha, iY_\alpha) = -2 (\alpha, H) (\mu_K(\psi), H_\alpha)$, while $\widetilde Q(iX_\alpha, iY_\alpha) = 0$.
  Since $(\mu_K(\psi), H_\alpha) > 0$, we conclude that the index of $Q$ is equal to twice the number of positive roots $\alpha$ with $(\alpha, H) > 0$.
\end{proof}


Since we have already seen above that the index of the Hessian is also equal to twice the complex dimension of $\calH(H < 0)$, we obtain the following result, which can also be extracted from \cite[Theorem 2]{Brion99}:

\begin{corollary}
\label{cor:what will be called trace condition}
  If $(-,H) \geq 0$ defines a non-trivial facet of the moment cone then
  \begin{equation*}
    \dim_{\CC} \calH(H < 0)
    = \# \{ \alpha \in R_{G,+} : (\alpha, H) > 0 \}
    = \# \{ \alpha \in R_{G,-} : (\alpha, H) < 0 \}.
  \end{equation*}
\end{corollary}

We now study the complexified group action.
To this end, we consider the Lie algebra $\mathfrak n_- = \bigoplus_{\alpha \in R_{G,-}} \mathfrak g_\alpha$ of the negative unipotent subgroup, which plays a role analogous to $\mathfrak r$ for vectors $\psi$ that are mapped into the positive Weyl chamber (compare the following with \cref{lem:torus complement injective}).

\begin{lemma}
  \label{lem:negative unipotent injective}
  The tangent map $\mathfrak n_- \rightarrow \calH, X \to \pi(X) \psi$ is injective.
\end{lemma}
\begin{proof}
  Let $E_- := \sum_{\alpha \in R_{G,+}} z_\alpha E_{-\alpha}$ be an arbitrary element in $\mathfrak n_-$.
  Since $\pi(E_{\pm\alpha})^\dagger = \pi(E_{\mp\alpha})$, we find that $\pi(E_-)^\dagger = \pi(E_+)$ with
  $E_+ := \sum_{\alpha \in R_{G,+}} \bar{z}_\alpha E_\alpha$ (the Cartan involution of $E_-$).
  Therefore,
  \begin{align*}
    \norm{\pi(E_-) \psi}^2
    &= \braket{\psi | \pi(E_+) \pi(E_-) | \psi}
    = \braket{\psi | \pi([E_+, E_-]) | \psi} + \norm{\pi(E_+) \psi}^2 \\
    &\geq \braket{\psi | \pi([E_+, E_-]) | \psi}.
  \end{align*}
  But \eqref{eq:onebody/sl2} implies that
  \begin{equation*}
    [E_+, E_-] - \sum_{\alpha \in R_{G,+}} \abs{z_\alpha}^2 H_\alpha \in i \mathfrak r,
  \end{equation*}
  so that by using $\mu_K(\psi) \in i \mathfrak t^*_{>0}$ we find that
  \begin{equation*}
    \norm{\pi(E_-) \psi}^2 \geq \braket{\psi | \pi([E_+, E_-]) | \psi} = \sum_{\alpha \in R_{G,+}} \abs{z_\alpha}^2 (\mu_K(\psi), H_\alpha) > 0.
    \qedhere
  \end{equation*}
\end{proof}

In contrast to \cref{lem:torus complement injective}, which continues to hold true if $\mu_K(\psi)$ is mapped to the relative interior of a different Weyl chamber, 
it is important in \cref{lem:negative unipotent injective} to choose the negative unipotent subgroup (relative to the choice of positive Weyl chamber).
For example, consider an irreducible $G$-representation $\calH = V_{G,\lambda}$ with highest weight $\lambda \in i\mathfrak t^*_{>0}$ and highest weight vector $v_\lambda$.
Then $\mu_K(v_\lambda) = \lambda \in i \mathfrak t^*_{>0}$ and the ``lowering operators'' in $\mathfrak n_-$ indeed act injectively.
On the other hand, the ``raising operators'' in the positive nilpotent Lie algebra $\mathfrak n_+$ annihilate the highest weight vector (by definition).

\medskip

We now decompose the Lie algebra $\mathfrak n_-$ similarly to \eqref{eq:weight space decomposition},
\begin{equation*}
  \mathfrak n_- = \mathfrak n_-(H < 0) \oplus \mathfrak n_-(H = 0) \oplus \mathfrak n_-(H > 0),
\end{equation*}
where $\mathfrak n_-(H < 0) = \bigoplus_{\alpha \in R_{G,-} : (\alpha,H) < 0} \mathfrak g_\alpha$ is the sum of the complex root spaces with negative $H$-weight $(\alpha, H) < 0$, etc.
We observe that \cref{cor:what will be called trace condition} can be equivalently stated as
\begin{equation}
\label{eq:dimensions match}
  \dim_\CC \calH(H < 0) = \dim_\CC \mathfrak n_-(H < 0).
\end{equation}
Note that $\pi(\mathfrak n_-(H < 0)) \calH(H = 0) \subseteq \calH(H < 0)$.
Thus we obtain the following important result:

\begin{proposition}
  \label{prp:non-trivial walls give isomorphism}
  Let $\psi \in \calH$ such that $\mu_K(\psi) \in i \mathfrak t^*_{>0}$ is a point on a non-trivial facet of the moment cone corresponding to the inequality $(-,H) \geq 0$.
  Then $\psi \in \calH(H = 0)$ and the tangent map restricts to an isomorphism
  \[ \mathfrak n_-(H < 0) \rightarrow \calH(H < 0), \quad X \mapsto \pi(X) \psi. \]
\end{proposition}
\begin{proof}
  The fact that $\psi \in \calH(H = 0)$ is just a reformulation of $\pi(H) \psi = 0$.
  By the preceding discussion, the tangent map is well-defined as a map from $\mathfrak n_-(H < 0)$ to $\calH(H < 0)$;
  it is injective by \cref{lem:negative unipotent injective} and surjective since the dimensions agree according to \eqref{eq:dimensions match}.
\end{proof}

We now prove a partial converse to \cref{prp:non-trivial walls give isomorphism}, inspired by the argument of Ressayre \cite{Ressayre10}.

\begin{proposition}
  \label{prp:partial converse}
  Suppose there exists $\psi \in \calH(H=0)$ such that the tangent map
  \begin{equation*}
    \mathfrak n_-(H < 0) \rightarrow \calH(H < 0), \quad
    X \mapsto \pi(X) \psi
  \end{equation*}
  is surjective. Then $(-, H) \geq 0$ is a valid inequality for the moment cone.
\end{proposition}
\begin{proof}
  Consider the smooth map
  \[N_- \times \calH(H \geq 0) \rightarrow \calH, \quad (g, \phi) \mapsto \Pi(g)\phi.\]
  Its differential at $(1,\psi)$ is the linear map
  \[\mathfrak n_- \oplus \calH(H \geq 0) \rightarrow \calH, \quad (X,V) \mapsto \pi(X) \psi + V.\]
  The assumption implies that this map is surjective.
  It follows that $\Pi(N_-) \calH(H \geq 0) \subseteq \calH$ contains a small Euclidean ball around $\psi$.
  In particular, any $N_-$-invariant polynomial that is zero on $\calH(H \geq 0)$ is automatically zero everywhere on $\calH$.

  We now prove the inequality.
  By the description of the moment cone in \eqref{eq:mumford description}, it suffices to show that
  $(\lambda, H) \geq 0$ for all highest weights $\lambda$ such that $V^*_{G,\lambda} \subseteq R(\calH)$.
  Recall that the highest weight of $V^*_{G,\lambda}$ is $\lambda^* = -w_0 \lambda$, where $w_0$ is the longest Weyl group element that flips the positive and negative roots.
  Consider a lowest weight vector, i.e., a polynomial $P \in R(\calH)$ that is a weight vector of weight $-\lambda$ and invariant under the action of $N_-$.
  Then $\pi(H) P = -(\lambda, H) P$ and the restriction of $P$ to $\calH(H \geq 0)$ is non-zero by our discussion above.
  But this restriction is an element of $R(\calH(H \geq 0)) = \Sym(\calH(H \geq 0))^*$, the space of polynomials on $\calH(H \geq 0)$.
  Since all $H$-weights in $R(\calH(H \geq 0))$ are non-positive, it follows that $-(\lambda, H) \leq 0$, as we set out to prove.
\end{proof}

We remark that \cref{prp:partial converse} holds unconditionally without any assumption on the dimension of the moment cone $C_K$.
We summarize our findings in the following definition and theorem that we had already advertised in the introduction.

\begin{definition}
\label{dfn:ressayre}
  An element $H \in i \mathfrak t$ is called a \emph{Ressayre element} if
  \begin{enumerate}
  \item $H$ is admissible, i.e., the linear hyperplane $(-, H) = 0$ is spanned by a subset of weights in $\Omega$, and
  \item there exists $\psi \in \calH(H = 0)$ such that the map
    \begin{equation}
    \label{eq:tangent map}
      \mathfrak n_-(H < 0) \rightarrow \calH(H < 0), \quad X \mapsto \pi(X) \psi
    \end{equation}
    is an isomorphism.
\end{enumerate}
\end{definition}

\begin{maintheorem}
  The moment cone is given by
  \begin{equation*}
    C_K = \{ \lambda \in i \mathfrak t^*_+ : (H, \lambda) \geq 0 \text{ for all Ressayre elements $H$} \}.
  \end{equation*}
\end{maintheorem}
\begin{proof}
  This follows directly from \cref{lem:admissible,prp:non-trivial walls give isomorphism,prp:partial converse}.
\end{proof}

\Cref{thm:main} gives a complete description of the moment cone of an arbitrary finite-dimensional $K$-representation $\calH$ (under the assumption that $C_K$ is of maximal dimension).
The set of inequalities thus obtained may still be redundant (i.e., not all inequalities necessarily correspond to facets of the moment cone).
In contrast, Ressayre's \emph{well-covering pairs} \cite{Ressayre10} characterize the facets of the moment polytope precisely.
In our language, his condition amounts to requiring that that the generic fiber of the map $N_- \times_{N_-(H \geq 0)} \calH(H \geq 0) \rightarrow \calH$ is a point
(as opposed to only requiring that the map be dominant).
Our characterization is also related to \cite[Theorem 2]{Brion99}, which uses algebraic geometry to characterize non-trivial faces of arbitrary codimension.
Unlike \cref{prp:partial converse}, it relies on an assumption about lower-dimensional moment polytopes, which can in principle be obtained recursively.

\section{Generalized Trace and Horn Condition}
\label{sec:trace and horn condition}

We now extract two useful necessary conditions that have to hold for any Ressayre element $H$, and therefore for any non-trivial facet of the moment polytope. We will see in \cref{subsec:trace horn horn} below that they are a generalization of the classical Horn inequalities, which justifies our terminology.

The first condition, which we call the \emph{trace condition} is the observation that the domain and range of the tangent map \eqref{eq:tangent map} necessarily have to agree, i.e.,
\begin{equation}
\label{eq:trace condition}
  \dim \mathfrak n_-(H < 0) = \dim \calH(H < 0).
\end{equation}
We note that the right-hand side of \eqref{eq:trace condition} is invariant under the action of the Weyl group.
This suggests that we first compute the \emph{dominant} admissible $H_0$ and then determine those $H \in W_K \cdot H_0$ which satisfy the trace condition \eqref{eq:trace condition}.
For this, we will need the following well-known lemma.

\begin{lemma}
  Let $H_0 \in i \mathfrak t_+ := \{ H_0 \in i \mathfrak t : (\alpha, H_0) \geq 0 ~ \forall \alpha \in R_{G,+} \}$ and $H \in W_K \cdot H_0$.
  Let $w \in W_K$ be the unique Weyl group element such that $H = w \cdot H_0$ and $w_0 w \cdot \alpha \in R_{G,+}$ for all $\alpha \in R_{G,+}$ with $(\alpha, H_0) = 0$. Then,
  \[ \ell(w_0 w) = \#\{ \alpha \in R_{G,+} : (\alpha, H) > 0 \}. \]
\end{lemma}
\begin{proof}
  Fix any $\rho_K \in i \mathfrak t_+$ such that $(\alpha, \rho_K) > 0$ for all positive roots $\alpha \in R_{G,+}$, 
  and set $H^\varepsilon := H - \varepsilon \rho_K$.
  Let $w \in W_K$ be a Weyl group element such that $H = w \cdot H_0$. 
  For all positive roots $\alpha \in R_{G,+}$ and $\varepsilon > 0$ small enough,
  \[ (\alpha, w^{-1} \cdot H^\varepsilon) = (\alpha, H_0) - \varepsilon (w \cdot \alpha, \rho_K) > 0 \]
  if and only if $(\alpha, H_0) = 0$ implies that $w \cdot \alpha \in R_{G,-}$.
  In other words, $H^\varepsilon_0 := w^{-1} \cdot H^\varepsilon \in i \mathfrak t_+$ if and only if $w_0 w \cdot \alpha \in R_{G,+}$ for all $\alpha \in R_{G,+}$ with $(\alpha, H_0) = 0$.
  Since $H^\varepsilon$ is regular, this immediately shows that such Weyl group elements $w$ exist and are unique.
  What is more, regularity also implies that
  \begin{align*}
      \ell(w_0 w)
    &= \ell(w_0 w^{-1})
    = \#\{ \alpha \in R_{G,+} : w^{-1} \cdot \alpha \in R_{G,+} \} \\
    &= \#\{ \alpha \in R_{G,+} : (w^{-1} \cdot \alpha, H^\varepsilon_0) > 0 \}
    = \#\{ \alpha \in R_{G,+} : (\alpha, w \cdot H^\varepsilon_0) > 0 \} \\
    &= \#\{ \alpha \in R_{G,+} : (\alpha, H^\varepsilon) > 0 \}
    = \#\{ \alpha \in R_{G,+} : (\alpha, H) > 0 \}. \qedhere
  \end{align*}
\end{proof}

We obtain the following useful corollary:

\begin{corollary}
\label{cor:trace criterion}
  Let $H_0 \in i \mathfrak t_+$ and $H \in W_K \cdot H_0$.
  Let $w \in W_K$ be the unique Weyl group element such that $H = w \cdot H_0$ and $w_0 w \cdot \alpha \in R_{G,+}$ for all $\alpha \in R_{G,+}$ with $(\alpha, H_0) = 0$.
  Then $H$ satisfies the trace condition \eqref{eq:trace condition} if and only if $\ell(w_0 w) = \dim_{\CC} \calH(H_0 < 0)$.
\end{corollary}

\medskip

The second condition, called the \emph{Horn condition}, is based on the observation that for any Ressayre element $H$ the \emph{determinant polynomial}
\begin{equation}
\label{eq:determinant}
  \delta_H \colon \begin{cases} \calH(H = 0) \rightarrow \CC \\
  \psi \mapsto \det \Big( \mathfrak n_-(H < 0) \ni X \mapsto \pi(X) \psi \in \calH(H < 0) \Big) \end{cases}
\end{equation}
is non-zero (we take the determinant with respect to any fixed pair of bases).
This can be understood to imply a statement about a smaller moment cone.
To see this, let $G(H=0)$ denote the identity component of the centralizer of the torus generated by $H$, i.e., the connected subgroup of $G$ with Lie algebra $\mathfrak g(H=0) = \mathfrak h \oplus \bigoplus_{\alpha:(\alpha,H)=0} \mathfrak g_\alpha$; denote by $N_-(H=0)$ the corresponding negative unipotent subgroup and by $K(H=0) = G(H=0) \cap K$ the maximal compact subgroup.
Then $G(H=0)$ acts on $\calH(H=0)$ and thus on the space $R(\calH(H=0)) = \Sym(\calH(H=0))^*$ of polynomial functions.

\begin{hornproposition}[Horn condition]
\hornstatement
\end{hornproposition}
\begin{proof}
  In view of \eqref{eq:mumford description}, it suffices to argue that the determinant polynomial $\delta_H$ is a lowest weight vector in $R(\calH(H=0))$ of weight $-\kappa_H$.


  To see this, fix a basis $\psi_1, \dots, \psi_k$ of $\calH(H < 0)$ such that each $\psi_j$ is a weight vector of weight $\omega_j$, and denote by $\alpha_1, \dots, \alpha_k$ the negative roots with $(\alpha, H) < 0$, so that $E_k := E_{\alpha_k}$ is a basis of $\mathfrak n_-(H < 0)$.
  Then the determinant polynomial $\delta_H$ with respect to this basis can be written as
  \[ \delta_H(\psi) = \braket{\psi_1 \wedge \dots \wedge \psi_k | \Lambda^k(\pi(-)\psi) | E_1 \wedge \dots \wedge E_k} \]
  where we write $\Lambda^k A$ for the canonical homomorphism $\Lambda^k V \rightarrow \Lambda^k W$ induced by a linear map $A \colon V \rightarrow W$.
  It follows that for any $g \in G$ and $\psi \in \calH(H=0)$,
  \begin{align*}
    &\quad (g \cdot \delta_H)(\psi)
    = \delta_H(\Pi(g^{-1}) \psi) \\
    &= \braket{\psi_1 \wedge \dots \wedge \psi_k | \Lambda^k(\Pi(g^{-1})) \, \Lambda^k(\pi(-)\psi) \, \Lambda^k(\Ad(g)) | E_1 \wedge \dots \wedge E_k}.
  \end{align*}
  Both $\Lambda^k(\mathfrak n_-(H < 0)) = \CC E_1 \wedge \dots \wedge E_k$ and $\Lambda^k(\calH(H < 0)) = \CC \psi_1 \wedge \dots \wedge \psi_k$ are one-dimensional representations of $K(H=0)$.
  In particular, they are invariant under the action of $N_-(H=0)$, and it is immediate that their weight is given by $\sum_k \alpha_k$ and $\sum_k \omega_k$, respectively.
  It follows at once that $\delta_H$ is a lowest weight vector of weight $-\kappa_H$.
\end{proof}

\section{Computation of Moment Cones}
\label{sec:computation}

\Cref{thm:main} reduces the computation of the moment cone of an arbitrary finite-dimensional representation to an enumeration of all Ressayre elements, which in principle is straighforward:
Since there are only finitely many weights, the admissibility condition cuts down the number of possible inequalities down to a finite list of candidates, and for each such candidate $(-,H) \geq 0$, the isomorphism condition can be easily checked.
Indeed, we only need to verify the trace condition \eqref{eq:trace condition} and that the determinant polynomial \eqref{eq:determinant} is non-zero.
In this way, we obtain a deterministic algorithm to compute the moment cone for an arbitrary representation that can easily be implemented in a computer program \cite{onlinecode}.
We remark that checking whether the determinant polynomial is non-zero can be sped up by using a fast \emph{probabilistic} algorithm for polynomial identity testing, e.g., based on the Schwartz-Zippel lemma.

In practice, naively enumerating all admissible hyperplanes by considering all $(r_K-1)$-element subsets of $\Omega$ quickly becomes infeasible as one considers representations of larger dimensions.
In this case, it is useful to first determine the admissible elements \emph{up to the Weyl group} and to impose further necessary conditions by a more refined analysis of the representation at hand.
Then \cref{cor:trace criterion} can be used to obtain directly only those candidates that satisfy the trace condition.
In this way, we may often cut down the number of candidates substantially for which we need to check that $\delta_H$ is non-zero.
In \cref{sec:qmp} below, we illustrate this analysis in the case of the one-body quantum marginal problem.

\section{Quantum Marginal Problem and Kronecker Cone}
\label{sec:qmp}

The state of a quantum system is specified by a unit vector $\psi$ in a Hilbert space $\calH$
, which we will always assume to be finite-dimensional.
Quantum systems composed of several distinguishable particles are described by the tensor product of the Hilbert spaces of their constituents, $\calH = \bigotimes_{k=1}^n \calH_k$.
The state of the $k$-th subsystem can be described by the \emph{reduced density matrix} $\rho_k$, which is the unique positive semi-definite operator on $\calH_k$ such that
\begin{equation}
\label{eq:rdm}
  \tr \rho_k X_k = \braket{\psi | \id^{\otimes (k-1)} \otimes X_k \otimes \id^{\otimes (n-k)} | \psi}
\end{equation}
for all Hermitian operators $X_k$ on $\calH_k$ (it is also called the \emph{partial trace} of $\psi$).
The fundamental \emph{one-body quantum marginal problem} asks which $\rho_1, \dots, \rho_n$ can arise as the reduced density matrices of a quantum state $\psi \in \calH$ \cite{Klyachko04,ChristandlMitchison06,DaftuarHayden04,ChristandlHarrowMitchison07,Walter14}.
Equivalently, it asks for the compatibility conditions that the $\rho_k$ have to satisfy in order for there to exist a global state $\psi$.
For two subsystems, it is a straightforward consequence of the singular value decomposition that $\rho_1$ and $\rho_2$ are compatible if and only if they have the same non-zero eigenvalues (including multiplicities).
In general, however, the problem is much more involved; it can be shown that it is a strict generalization of the problem of computing the Horn cone of \cref{sec:horn} \cite{Klyachko04,ChristandlSahinogluWalter12,Walter14} and it has been solved in \cite{Klyachko04,DaftuarHayden04} by using similar methods \cite{Klyachko98,BerensteinSjamaar00}. However, a concrete description akin to the Horn inequalities is still elusive.

In this section, we shall consider the case of three subsystems, to which the general case can always be reduced \cite{Walter14}.
By comparing \eqref{eq:rdm} and \eqref{eq:moment map}, it is not hard to see that solving the one-body quantum marginal problem amounts to computing the moment cone for the action of $K = \SU(a) \times \SU(b) \times \SU(c) \times \U(1)$ on $\calH = \CC^a \otimes \CC^b \otimes \CC^c$.
The rank of $K$ is $r = (a-1) + (b-1) + (c-1) + 1 = a+b+c-2$.
The elements of $i\mathfrak t$ can be identified with quadruples $H = (H_A, H_B, H_C, z) \in \RR^{a+b+c+1}$ with $\sum_i H_{A,i} = \sum_j H_{B,j} = \sum_k H_{C,k} = 0$.
It will be convenient to identify the elements of the dual space $i \mathfrak t^*$ with triples $\lambda = (\lambda_A, \lambda_B, \lambda_C) \in \RR^{a+b+c}$ with $\abs\lambda := \sum_i \lambda_{A,i} = \sum_j \lambda_{B,j} = \sum_k \lambda_{C,k}$. The natural pairing of elements $H \in i\mathfrak t$ and $\lambda \in i\mathfrak t^*$ is then given by
\[
  (H,\lambda) = (H_A, \lambda_A) + (H_B, \lambda_B) + (H_C, \lambda_C) + z \abs\lambda.
\]
Thus the set of weights of $\calH$ is 
$
  \Omega = \{ (e_{A,i}, e_{B,j}, e_{C,k}) : i = 1,\dots,a, \, j=1,\dots,b, \, k=1,\dots,c \},
$
where we write $e_{A,i}$ for the vector in $\RR^a$ with $i$-th entry equal to $1$ and all other entries equal to zero, etc.
Throughout this section, we will denote the moment cone by $C(a,b,c)$.

By Schur-Weyl duality and Mumford's description \eqref{eq:mumford description}, the moment cone $C(a,b,c)$ can equivalently be defined in terms of the representation theory of the symmetric group \cite{ChristandlMitchison06,Klyachko04,ChristandlHarrowMitchison07,ChristandlDoranKousidisWalter2014}: We have
\[ C(a,b,c) = \cone \{ (\alpha,\beta,\gamma) : g_{\alpha,\beta,\gamma} \neq 0 \}, \] 
where $\alpha$, $\beta$, and $\gamma$ vary over the set of Young diagrams with the same number $k$ of boxes and no more than $a$, $b$, and $c$ rows, respectively, and where $g_{\alpha,\beta,\gamma}$ denotes the \emph{Kronecker coefficient of the symmetric group}, i.e., the multiplicity of the invariant subspace in the corresponding triple tensor product of irreducible $S_k-$representations.
We will henceforth refer to $C(a,b,c)$ as the \emph{Kronecker cone}.

It will be convenient to assume without loss of generality that $1 < a \leq b \leq c \leq ab$.
In this case, $C(a,b,c)$ is maximal-dimensional (see \cref{cor:maxdim} in the \refappendix) and our method is directly applicable.
In fact, we shall see below that all other cases can be reduced to the case $c=ab$.
According to \cref{thm:main}, the moment cone $C(a,b,c)$ is thus cut out by those $H$ which are Ressayre elements, i.e., which are admissible and satisfy the isomorphism condition.
Naively determining the admissible $H$ by enumerating all subsets of $\Omega$ with cardinality $r-1$, determining whether they span a hyperplane and computing the normal vector amounts to considering $\binom {\abs \Omega} {r-1} = \binom {abc} {a+b+c-3}$ subsets, which rapidly becomes infeasible (e.g., for $a=b=c=4$, there are over 27 billion such subsets). 
We therefore need to derive additional constraints to make this approach computationally feasible.

\subsection{Candidates}
\label{subsec:qmp candidates}

As a first step, we recall that the bipartite variant of the problem has a straight-forward solution: The moment cone of $\SU(a) \times \SU(b) \times \U(1)$ acting on $\CC^a \otimes \CC^b$ is given by
\begin{equation}
\label{eq:bipartite cone}
\begin{aligned}
  C(a, b) = \{ (\lambda_A, \lambda_B) \in \RR^{a+b} : &\lambda_{A,1} \geq \dots \geq \lambda_{A,a} \geq 0, \\
   &\lambda_B = (\lambda_A, 0, \dots, 0) \},
\end{aligned}
\end{equation}
where we have used the same conventions as above and assumed that $a \leq b$.
This can be easily proved directly; it also follows from Mumford's description of the moment polytope, since $\Sym^k(\CC^a \otimes \CC^b) = \bigoplus_\lambda V^a_\lambda \otimes V^b_\lambda$ by Schur-Weyl duality, where the direct sum is over all Young diagrams $\lambda$ with $k$ boxes and no more than $a = \min(a,b)$ rows.
In quantum-mechanical terms, \eqref{eq:bipartite cone} amounts to the well-known assertion that the reduced density matrices of a bipartite pure state are isospectral.

By considering $\SU(a) \times \SU(b) \times \SU(c) \subseteq \SU(ab) \times \SU(c)$, the moment cone for the tripartite problem can now be written in terms of the bipartite cone $C(ab, c)$ and the moment polytopes $\Delta(a,b|\lambda_{AB})$ of the action of $\SU(a) \times \SU(b)$ on the coadjoint $\SU(ab)$-orbits $\calO^{ab}_{\lambda_{AB}}$:
\begin{equation}
\label{eq:tripartite via bipartite conditioning}
\begin{aligned}
  C(a, b, c) = \{ (\lambda_A, \lambda_B, \lambda_C) ~:~~ &\exists \lambda_{AB} \text{ s.th.} (\lambda_{AB}, \lambda_C) \in C(ab, c), \\
  &(\lambda_A, \lambda_B) \in \Delta(a,b|\lambda_{AB}) \}
\end{aligned}
\end{equation}
A first consequence is that the case $c \neq ab$ can always be reduced to $c = ab$. Indeed, \eqref{eq:bipartite cone} and \eqref{eq:tripartite via bipartite conditioning} imply that
\begin{align}
\nonumber
  C(a, b, c) &= \{ (\lambda_A, \lambda_B, (\lambda_{AB}, 0, \dots, 0)) : (\lambda_A, \lambda_B, \lambda_{AB}) \in C(a, b, ab) \} \\
\intertext{if $c > ab$, and}
\label{eq:tripartite cone}
  C(a, b, c) &= \{ (\lambda_A, \lambda_B, \lambda_C) : (\lambda_A, \lambda_B, (\lambda_C,0, \dots, 0)) \in C(a, b, ab) \}
\end{align}
if $c < ab$.
Thus the moment cone for $c > ab$ is isometric to $C(a,b,ab)$, while for $c < ab$ it is obtained as a projection of the latter.
In the case where $c = ab$, it was observed in \cite{Manivel97,Klyachko04} that any normal vector of the moment cone $C(a, b, ab)$ necessarily has a rather special form.
We state this result and give a succint alternative proof that does not rely on the results of \cite{Klyachko04,BerensteinSjamaar00}:

\begin{lemma}
\label{lem:bipartite klyachko}
  Let $H = (H_A, H_B, H_{AB}, z)$ be the normal vector of a non-trivial facet of the moment cone $C(a, b, ab)$ such that $(H_A, H_B) \neq 0$.
  Then:
  \begin{enumerate}
    \item \label{item:extremal edge} $(H_A, H_B)$ is determined by a maximal number of equations of the form $H_{A,i} + H_{B,j} = H_{A,k} + H_{B,l}$,
    \item the components of $H_{AB}$ are precisely all possible partial sums $-H_{A,i}-H_{B,j}$, and
    \item $z = 0$.
  \end{enumerate}
\end{lemma}
\begin{proof}
  By \cref{lem:non-trivial facets}, there exists a regular dominant point $\lambda = (\lambda_A, \lambda_B, \lambda_{AB}) \in i \mathfrak t^*_{>0}$ in the interior of this facet.
  In a neighborhood of $\lambda$, the moment cone locally looks like a half-space, so that
  \begin{equation}
    \label{eq:bipartite facet}
    (H_A, \mu_A) + (H_B, \mu_B) \geq z \abs\lambda - (H_{AB}, \lambda_{AB})
  \end{equation}
  is not only a valid inequality that holds for all $(\mu_A, \mu_B) \in \Delta(a,b|\lambda_{AB})$, as follows from \eqref{eq:tripartite via bipartite conditioning}, but in fact a facet of the moment polytope $\Delta(a,b|\lambda_{AB})$, since we have assumed that $(H_A, H_B) \neq 0$.

  While determining the moment polytopes $\Delta(a,b|\lambda_{AB})$ is just as hard as determining the cone $C(a, b, ab)$, this reformulation gives us an additional insight:
  Recall that the Duistermaat-Heckman measure for the action of the maximal torus $T(a) \times T(b)$ of $\SU(a) \times \SU(b)$ on the coadjoint $\SU(ab)$-orbit $\calO^{ab}_{\lambda_{AB}}$ is a measure on the abelian moment polytope.
  Now let $\pi$ denote the projection $i \mathfrak t(ab)^* \rightarrow i \mathfrak t(a)^* \oplus i \mathfrak t(b)^*$.
  The affine hyperplanes through some $\pi(w \lambda_{AB})$ spanned by subsets of the \emph{restricted roots} $\pi(\alpha)$, $\alpha \in R_{ab,+}$, partition the abelian moment polytope into a finite number of polyhedral chambers,
  and it as an immediate consequence of the Heckman formula \cite{Harish-Chandr57, Heckman82} that the measure has a polynomial density function on each chamber (cf.~\cite{BoysalVergne09}).
  The Duistermaat-Heckman measure for the action of $\SU(a) \times \SU(b)$ can be recovered by applying a number of partial derivatives to the measure for $T(a) \times T(b)$ (e.g., \cite{ChristandlDoranKousidisWalter2014}).
  It follows that the moment polytope $\Delta(a,b|\lambda_{AB})$, which is the support of the latter measure, is equal to a finite union of chambers; in particular, its non-trivial facets are contained in hyperplanes of the form just described.

  Applied to the facet \eqref{eq:bipartite facet}, it follows that its normal vector $(H_A, H_B)$ is defined by a maximal number of equations of the form
  \begin{align*}
    ((H_A, H_B), \pi(\alpha)) &= (H_A \otimes \id_B + \id_A \otimes H_B, \alpha) \\
    &= H_{A,i} + H_{B,j} - H_{A,k} - H_{B,l} = 0
  \end{align*}
  for some indices $i, j, k, l$. This shows the first assertion.
  Moreover, since the facet contains $\pi(w \lambda_{AB})$ for some $w \in S_{ab}$ we obtain that
  \begin{align*}
    ((H_A, H_B), \pi(w \lambda_{AB})) &= (w^{-1}(H_A \otimes \id_B + \id_A \otimes H_B), \lambda_{AB}) \\
    &= z \abs \lambda - (H_{AB}, \lambda_{AB}).
  \end{align*}
  A priori, the permutation $w$ will depend on the choice of $\lambda_{AB}$.
  However, there are only finitely many permutations $w$, while we may vary $\lambda_{AB}$ arbitrarily in a small neighborhood.
  It follows that $w^{-1}(H_A \otimes \id_B + \id_A \otimes H_B) = -H_{AB}$ and $z = 0$ (since $H_A \otimes \id_B + \id_A \otimes H_B$ is traceless).
  This shows the second and third assertion.
\end{proof}

Following Klyachko, we shall call any $(H_A, H_B)$ that satisfies condition \ref{item:extremal edge} of \cref{lem:bipartite klyachko} an \emph{extremal edge} if it is in addition dominant and primitive (in the dual of the root lattice).
There are only finitely many extremal edges and we shall denote them by $\mathcal E_+(a, b)$.
The extremal edges span the extreme rays of the \emph{cubicles}, which are the full-dimensional convex cones of elements $(H_A,H_B)$ cut out by a maximal set of inequalities of the form $H_{A,i} + H_{B,j} \geq H_{A,k} + H_{B,l}$.
In other words, a cubicle is defined as a set of $(H_A,H_B)$ with fixed order of the $H_{A,i} + H_{B,j}$ and can therefore be encoded by a standard Young tableaux of rectangular shape $a \times b$ \cite{Klyachko04}.
This gives a straightforward way of computationally determining all extremal edges (see \cref{tab:bipartite admissibility}).

We remark that standard tableaux that correspond to cubicles have also been called \emph{additive} in the literature \cite{Vallejo14,Manivel14}.
Manivel has shown that they can be associated with minimal regular faces of the moment cone for which the corresponding Kronecker coefficients stabilize \cite{Manivel97,Manivel14}.
The corresponding extremal edges determine non-trivial facets of $C(a,b,ab)$ incident to this face \cite{Manivel97,Klyachko04}.
However, not all non-trivial facets can be obtained in this way.

\begin{table}
\begin{center}
  \begin{tabular}{lrrr}
    \toprule
    \multicolumn{1}{c}{$(a, b)$} & \multicolumn{1}{c}{$(2, 2)$ \cite{Klyachko04}} & \multicolumn{1}{c}{$(3, 3)$ \cite{Klyachko04}} & \multicolumn{1}{c}{$(4, 4)$} \\
\midrule
Tableaux & 2 & 42 & 24024 \\
Cubicles & 2 & 36 & 6660 \\
Extremal edges ($\abs{\mathcal E_+}$) & 3 (2) & 17 (10) & 457 (233) \\
\bottomrule
  \end{tabular}
\end{center}
  \caption{Bipartite candidates (counts in parentheses are with permutations removed)}
  \label{tab:bipartite admissibility}
\end{table}

\begin{corollary}
\label{cor:bipartite admissibility}
  Let $H = (H_A, H_B, H_C, z)$ be the normal vector of a non-trivial facet of the moment cone $C(a,b,c)$, where $c \leq ab$.
  Then $(H_A,H_B) = 0$ or $(H_A,H_B)$ is proportional to an element in the $S_a \times S_b$-orbit of $\mathcal E_+(a, b)$ (i.e., an extremal edge up to $S_a \times S_b$ and rescaling).
\end{corollary}
\begin{proof}
  In view of \eqref{eq:tripartite cone}, we may obtain finite and complete set of inequalities for $C(a,b,c)$ by taking any facet $H = (H_A, H_B, H_{AB}, 0)$ for $C(a,b,ab)$ and restricting $H_{AB}$ to its first $c$ components (that is, set $H_C$ to be the traceless part of $(H_{AB,1},\dots,H_{AB,c})$ and define $z$ accordingly).
  Such an inequality will not necessarily define a facet of $C(a,b,c)$, but all facets arise in this way.
  If we started with a non-trivial inequality then the claim follows from \cref{lem:bipartite klyachko}.
  If we started with a trivial inequality then we either obtain a trivial inequality or $(H_A,H_B) = 0$.
\end{proof}

\Cref{lem:bipartite klyachko} is efficient in finding candidates for facets of $C(a, b, ab)$, but not necessarily so for $C(a, b, c)$with $c < ab$. Indeed, while we know from the proof of \cref{cor:bipartite admissibility} that any facet of the latter can be obtained by ``restriction'' of a facet of the former, for each given facet there are many possible restrictions, as we need to pick a subset of $c$ components out of the $ab$ components of the given $H_{AB}$.
If $c \ll ab$ then it can be more efficient to apply \cref{cor:bipartite admissibility} to all three of $(H_A,H_B)$, $(H_A,H_C)$ and $(H_B, H_C)$ (e.g., in the case of $a=b=c=4$).
For this we will use the following lemma:


\begin{lemma}
\label{lem:primitivity}
  Let $(H_A, H_B)$ be an extremal edge.
  Then $H_A$ and $H_B$ are each either zero or primitive in the dual of the root lattice.
\end{lemma}
\begin{proof}
  By symmetry, it suffices to show assertion for $H_A$.
  Let $\Omega(a, b) := \{(\alpha, \beta) \neq 0 : \alpha \in R_a \cup \{0\}, \beta \in R_b \cup \{0\} \}$ and denote by $S \subseteq \Omega(a, b)$ the subset of all elements that are orthogonal to $H := (H_A, H_B)$.
  Note that the linear hyperplane $H^\perp$ is always spanned by the set of $S$ (this is just a reformulation of the defining property of an extremal edge).
  Define $S_B := \{ \beta : (\alpha, \beta) \in S \}$ and consider the subspace $\mathfrak m := \Span_{\RR} S_B \subseteq i \mathfrak t(b)$.

  Case 1: $\mathfrak m \subsetneq i \mathfrak t(b)$.
  Then $H^\perp = \Span_{\RR} S \subseteq i \mathfrak t(a) \oplus \mathfrak m$.
  By comparing dimensions, it follows that we have in fact equality, $H^\perp = i \mathfrak t(a) \oplus \mathfrak m$.
  We conclude that $H_A = 0$.

  Case 2: $\mathfrak m = i \mathfrak t(b)$.
  Since the matrix with columns the roots of $A_{b-1}$ is totally unimodular (e.g., \cite[p.~274, (18)]{Schrijver86}), it follows that $S_B$ spans the root lattice.
  To find a contradiction, suppose that $H_A$ is neither zero nor primitive.
  Then we can write $H_A = n H'_A$ for some $n > 1$ and a non-zero element $H'_A$ in the dual of the root lattice.
  But then,
  \[ (H_B, \beta) = -(H_A, \alpha) = -n(H'_A, \alpha_i) \in n \ZZ \]
  for all $(\alpha, \beta) \in S$.
  Since the $S_B$ spans the root lattice, we conclude that also $H_B / n$ is an element of the dual of the root lattice.
  But then $H = (H_A, H_B)$ is not primitive, which is the desired contradiction.
\end{proof}

\begin{lemma}
\label{lem:tripartite admissibility}
  Let $H = (H_A, H_B, H_C, z)$ be the normal vector of a non-trivial facet of the moment cone $C(a,b,c)$, where $c \leq ab$.
  Then $(H_A, H_B, H_C)$ is proportional to an element in the $S_a \times S_b \times S_c$-orbit of
  \begin{align*}
    \mathcal E_+(a, b, c) := \{ &(H_A, H_B, H_C) \neq 0 :
    (H_A, H_B) \in \mathcal E_+(a, b) \cup \{0\} , \, \\
      &(H_A, H_C) \in \mathcal E_+(b, c) \cup \{0\}, \,
      (H_B, H_C) \in \mathcal E_+(a, c) \cup \{0\}
    \}.
  \end{align*}
\end{lemma}
\begin{proof}
  Since the following argument is equivariant under the Weyl group $S_a \times S_b \times S_c$, we may without loss of generality assume that $(H_A, H_B, H_C)$ is dominant

  Case 1: Only on component is non-zero, say, $H_C \neq 0$ and therefore $(H_A, H_B) = 0$.
  We may rescale $H$ such that $H_C$ is primitive. Then $(H_A, H_C)$ and $(H_B, H_C)$ are also primitive and therefore extremal edges by \cref{cor:bipartite admissibility}.

  Case 2: At least two components are non-zero, say, $H_A$ and $H_B$.
  We may rescale $H$ such that $(H_A, H_B)$ is primitive.
  Since $H_A$ and $H_B$ are non-zero, \cref{lem:primitivity} shows that both $H_A$ and $H_B$ are individually primitive.
  It follows that three of $(H_A, H_B)$, $(H_A, H_C)$ and $(H_B, H_C)$ are primitive and therefore extremal edges by \cref{cor:bipartite admissibility}.
\end{proof}

By \cref{thm:main}, any non-trivial facet is necessarily admissible.
As admissibility is a Weyl group-invariant property, we immediately obtain the following corollary:

\begin{corollary}
\label{cor:tripartite admissibility}
  Let $H$ be the normal vector of a non-trivial facet of the moment cone $C(a,b,c)$, where $c \leq ab$.
  Then $H$ is proportional to an element in the $S_a \times S_b \times S_c$-orbit of
  \begin{align*}
    \mathcal E_{+,\text{adm}}(a, b, c) := \{ H = (H_A, H_B, H_C, z) \neq 0 :
      &(H_A, H_B, H_C) \in \mathcal E_+(a, b, c), \\
      &H \text{ admissible} \}.
  \end{align*}
\end{corollary}

For any given $(H_A, H_B, H_C) \in \mathcal E_+(a, b, c)$, there are in general several $z$ such that $H = (H_A, H_B, H_C, z)$ is admissible (or no such $z$ at all). The sets $\mathcal E_+(a, b, c)$ and $\mathcal E_{+,\text{adm}}$ consist of dominant, primitive elements, and they can be easily obtained algorithmically from the sets of extremal edges.

As a final step we implement the trace condition \eqref{eq:trace condition}.
For this, recall that the \emph{length} of a permutation $\pi$ is equal to the number of inversions, $\ell(\pi) = \#\{ (i,j) : i < j, \pi(i) > \pi(j) \}$.
A permutation $\pi$ is called a \emph{shuffle} with respect to a dominant vector $x \in \RR^d$ if $x_i = x_j$ for $i < j$ implies that $\pi(i) < \pi(j)$.
Then the following is an immediate consequence of \cref{cor:tripartite admissibility,cor:trace criterion}:

\begin{corollary}
\label{cor:tripartite candidates}
  Let $H$ be the normal vector of a non-trivial facet of the moment cone $C(a,b,c)$, where $c \leq ab$.
  Then $H$ is proportional to an element in
  \[
    \mathcal E(a, b, c) := \{ H = w \cdot H_0 : H_0 \in \mathcal E_{+,\text{adm}}, w \in W(H_0) \}
   \]
  where $W(H_0)$ denotes the set of triples $w = (w_A,w_B,w_C) \in S_a \times S_b \times S_c$ such that $(w_0 w_A, w_0 w_B, w_0 w_C)$ is a triple of shuffles with respect to the components of $H_0$ whose lengths sum up to $\dim \calH(H_0 < 0)$.
\end{corollary}

We remark that it is straightforward to computationally generate all shuffles of a given length by adapting the algorithm of Effler and Ruskey \cite{EfflerRuskey03}.
The upshot of \cref{cor:tripartite candidates} then is that the moment cone $C(a, b, c)$ is cut out by those candidates in $\mathcal E(a, b, c)$ that are also Ressayre elements (together with the trivial inequalities):
\[ C(a, b, c) = \{ \lambda \in i \mathfrak t^*_+ : (H, \lambda) \geq 0 \text{ for all Ressayre elements $H \in \mathcal E(a, b, c)$} \}. \]
We conclude this section with an illustrative example of the method.

\begin{example}
\label{ex:polygonal}
  Consider the moment cone $C(d,d,d)$ for any $d > 1$.
  It is not hard to verify that all pairs formed from $(1,\dots,1,1-d)$ and $(d-1,-1,\dots,-1)$ are extremal edges 
  (formally, we should divide by $d$ to work with primitive elements, but we refrain from doing so in the interest of readability).
  Therefore,
  \[ \big((1,\dots,1,1-d), (1,\dots,1,1-d), (d-1,-1,\dots,-1)\big) \in \mathcal E_+(d,d,d). \]
  There are two possible values of $z$ which extend the above to an admissible element $H_0 \in \mathcal E_{+,\text{adm}}(d,d,d)$.
  The first option is $z=-1$. However, the dimension of $\calH(H_0 < 0)$ is $2(d-1)^2 + d$, and therefore strictly larger than the dimension of $\mathfrak n_-$. Thus the trace condition \eqref{eq:trace condition} can never be satisfied!

  The second option is $z=d-1$. Here we find that $\dim \calH(H_0 < 0) = d-1$.
  Consider the Weyl group element $w$ such that $w_0 w = (\id,\id,\sigma)$ where $\sigma$ is the permutation that sends $1 \mapsto d$ and all other $k \mapsto k-1$.
  Observe that $w$ is an element in the set $W(H_0)$ defined in \cref{cor:tripartite candidates}, so that we obtain the following candidate:
  \begin{align*}
    H &= w \cdot H_0 = \big( (1-d,1,\dots,1), (1-d,1,\dots,1), (d-1,-1,\dots,-1), d-1 \big) \\
    &\in \mathcal E(d,d,d).
  \end{align*}
  We now verify that $H$ is a Ressayre element. For this, observe that we have
  \begin{align*}
    \mathfrak n_-(H < 0) &= \Span \{ (0,0,E_{\alpha_{21}}), \dots, (0,0,E_{\alpha_{d1}}) \} \\
    \calH(H = 0) &= \Span \{ e_{111}, e_{i1k}, e_{1jk} : i,j,k=2,\dots,d \} \\
    \calH(H < 0) &= \Span \{ e_{112}, \dots, e_{11d} \},
  \end{align*}
  where $e_{ijk} := e_i \otimes e_j \otimes e_k$, with $e_1, \dots, e_d$ the standard basis in $\CC^d$.
  We note that all basis vectors in $\calH(H=0)$ are annihilated by $\mathfrak n_-(H<0)$ except for $e_{111}$, which is sent by the lowering operator $(0,0,E_{\alpha_{k1}})$ to $e_{11k}$ ($k=2,\dots,d$).
  Thus the tangent map \eqref{eq:tangent map} is diagonal with respect to the basis given above and the determinant polynomial \eqref{eq:determinant} is given by
  \[ \delta_H(\psi) = \det \left(\begin{smallmatrix} \psi_{111} & & \\ & \ddots & \\ & & \psi_{111} \end{smallmatrix}\right) = \psi_{111}^{d-1} \neq 0, \]
  for any $\psi = \sum_{ijk} \psi_{ijk} e_{ijk} \in \calH(H=0)$.
  It follows at once from \cref{thm:main} that $(H,\lambda) \geq 0$ is a valid inequality for the moment polytope.
  We have thus obtained the well-known \emph{polygonal inequality} \cite{HiguchiSudberySzulc03}
  \[ \lambda_{A,1} + \lambda_{B,1} \leq \abs\lambda + \lambda_{C,1}, \]
  in a completely mechanical fashion.
  By symmetry, the two other inequalities obtained by permuting the subsystems $A$, $B$, and $C$ are also valid.
  We remark that the moment cone $C_{K(H=0)}(\calH(H=0))$ is closely related to the Horn cone for $\U(d-1)$ discussed in \cref{sec:horn} below \cite{ChristandlSahinogluWalter12,Walter14,Littlewood58,Murnaghan55}.
\end{example}

\subsection{Computational Results}

To verify that a given $H$ is a Ressayre element, we have implemented a computer program that works for arbitrary representations.
To compute the set of candidates $\mathcal E(a, b, c)$ in the case of the one-body quantum marginal problem, we have used the strategy explained above.
In \cref{tab:tripartite results} we list some results obtained by our program in the symmetric scenario $a=b=c$, which corresponds to three quantum particles with the same number of degrees of freedom.
While the moment cones $C(2,2,2)$ and $C(3,3,3)$ had already been computed in \cite{HiguchiSudberySzulc03,Bravyi04,Franz02,Higuchi03} using different methods, the cone $C(4,4,4)$ had been out of reach using current methods.
In contrast, our method allows us the computation of $C(4,4,4)$ in a few minutes, since it does not rely on an intermediate computation of the higher-dimensional cone $C(4,4,16)$.

\begin{table}
\begin{center}
  \begin{tabular}{lrrr}
    \toprule
    \multicolumn{1}{c}{$(a, b, c)$} & \multicolumn{1}{c}{$(2, 2, 2)$ \cite{HiguchiSudberySzulc03}} & \multicolumn{1}{c}{$(3, 3, 3)$ \cite{Franz02}} & \multicolumn{1}{c}{$(4, 4, 4)$} \\
\midrule
$\abs{\mathcal E_+}$ & 7 (3) & 51 (17) & 3027 (600) \\
$\abs{\mathcal E_{+,\text{adm}}}$ & 11 (5) & 67 (25) & 2231 (484) \\
$\abs{\mathcal E}$ & 9 (3) & 192 (41) & 32406 (5633) \\
\midrule
Inequalities & 9 (3) & 114 (25) & 1749 (323) \\
Facets & 6 (2) & 45 (10) & 270 (50) \\
Extreme Rays & 5 (3) & 33 (11) & 328 (65) \\
\bottomrule
  \end{tabular}
\end{center}
  \caption{Tripartite candidates (counts in parentheses are with permutations removed)}
  \label{tab:tripartite results}
\end{table}

In \cref{tab:fourfourfour rays,tab:fourfourfour facets}, we list the extreme rays and facets of $C(4, 4, 4)$ up to permutations of the subsystems.
Facet No.~21 is an instance of the polygonal inequality discussed in \cref{ex:polygonal}.
Illustrating a pattern observed in \cite{Franz02} for $C(3,3,3)$, there are several facets that contain neither the highest weight $((1,0,0,0),(1,0,0,0),(1,0,0,0))$ nor the ``origin'' $(\tau_4,\tau_4,\tau_4)$, where $\tau_d = \id/d \in \RR^d$.
Moreover, the only facets that contain the origin are the Weyl chamber inequalities.
This is an instance of a general fact that might be of independent interest (\cref{lem:origin ball} in the \refappendix).


\begin{table}
  \tiny
  \begin{center}
  \begin{tabular}{rlll}
    \toprule
    \multicolumn{1}{c}{$\#$} & \multicolumn{1}{c}{$V_A$} & \multicolumn{1}{c}{$V_B$} & \multicolumn{1}{c}{$V_C$} \\
\midrule
1 & $(1/4,1/4,1/4,1/4)$ & $(1/4,1/4,1/4,1/4)$ & $(1/4,1/4,1/4,1/4)$ \\
2 & $(1/4,1/4,1/4,1/4)$ & $(1/4,1/4,1/4,1/4)$ & $(1/3,1/3,1/3,0)$ \\
3 & $(1/4,1/4,1/4,1/4)$ & $(1/4,1/4,1/4,1/4)$ & $(1/2,1/2,0,0)$ \\
4 & $(1/4,1/4,1/4,1/4)$ & $(1/4,1/4,1/4,1/4)$ & $(1,0,0,0)$ \\
5 & $(1/4,1/4,1/4,1/4)$ & $(1/3,1/3,1/3,0)$ & $(1/3,1/3,1/3,0)$ \\
6 & $(1/4,1/4,1/4,1/4)$ & $(1/3,1/3,1/3,0)$ & $(1/2,1/2,0,0)$ \\
7 & $(1/4,1/4,1/4,1/4)$ & $(1/3,1/3,1/3,0)$ & $(2/3,1/6,1/6,0)$ \\
8 & $(1/4,1/4,1/4,1/4)$ & $(1/3,1/3,1/3,0)$ & $(2/3,1/4,1/12,0)$ \\
9 & $(1/4,1/4,1/4,1/4)$ & $(1/3,1/3,1/3,0)$ & $(3/4,1/12,1/12,1/12)$ \\
10 & $(1/4,1/4,1/4,1/4)$ & $(3/8,3/8,1/4,0)$ & $(5/8,3/8,0,0)$ \\
11 & $(1/4,1/4,1/4,1/4)$ & $(3/8,3/8,1/4,0)$ & $(3/4,1/8,1/8,0)$ \\
12 & $(1/4,1/4,1/4,1/4)$ & $(2/5,3/10,3/10,0)$ & $(7/10,3/20,3/20,0)$ \\
13 & $(1/4,1/4,1/4,1/4)$ & $(5/12,5/12,1/6,0)$ & $(2/3,1/6,1/12,1/12)$ \\
14 & $(1/4,1/4,1/4,1/4)$ & $(1/2,1/6,1/6,1/6)$ & $(1/2,1/2,0,0)$ \\
15 & $(1/4,1/4,1/4,1/4)$ & $(1/2,1/4,1/8,1/8)$ & $(5/8,3/8,0,0)$ \\
16 & $(1/4,1/4,1/4,1/4)$ & $(1/2,1/4,1/4,0)$ & $(1/2,1/2,0,0)$ \\
17 & $(1/4,1/4,1/4,1/4)$ & $(1/2,1/4,1/4,0)$ & $(2/3,1/6,1/6,0)$ \\
18 & $(1/4,1/4,1/4,1/4)$ & $(1/2,1/4,1/4,0)$ & $(3/4,1/4,0,0)$ \\
19 & $(1/4,1/4,1/4,1/4)$ & $(1/2,3/8,1/8,0)$ & $(5/8,1/8,1/8,1/8)$ \\
20 & $(1/4,1/4,1/4,1/4)$ & $(1/2,1/2,0,0)$ & $(1/2,1/2,0,0)$ \\
21 & $(2/7,2/7,2/7,1/7)$ & $(4/7,1/7,1/7,1/7)$ & $(4/7,3/7,0,0)$ \\
22 & $(7/24,7/24,5/24,5/24)$ & $(1/3,1/3,1/3,0)$ & $(3/4,1/8,1/8,0)$ \\
23 & $(3/10,3/10,1/5,1/5)$ & $(2/5,3/10,3/10,0)$ & $(4/5,1/10,1/10,0)$ \\
24 & $(3/10,3/10,3/10,1/10)$ & $(1/2,1/2,0,0)$ & $(11/20,3/20,3/20,3/20)$ \\
25 & $(3/10,3/10,3/10,1/10)$ & $(1/2,1/2,0,0)$ & $(3/5,1/5,1/10,1/10)$ \\
26 & $(1/3,2/9,2/9,2/9)$ & $(1/3,1/3,1/3,0)$ & $(2/3,1/3,0,0)$ \\
27 & $(1/3,2/9,2/9,2/9)$ & $(1/3,1/3,1/3,0)$ & $(7/9,1/9,1/9,0)$ \\
28 & $(1/3,2/9,2/9,2/9)$ & $(4/9,4/9,1/9,0)$ & $(2/3,1/9,1/9,1/9)$ \\
29 & $(1/3,1/3,1/6,1/6)$ & $(1/3,1/3,1/3,0)$ & $(7/9,1/9,1/9,0)$ \\
30 & $(1/3,1/3,1/6,1/6)$ & $(1/3,1/3,1/3,0)$ & $(5/6,1/6,0,0)$ \\
31 & $(1/3,1/3,1/6,1/6)$ & $(1/2,1/4,1/4,0)$ & $(3/4,1/12,1/12,1/12)$ \\
32 & $(1/3,1/3,1/6,1/6)$ & $(2/3,1/6,1/6,0)$ & $(2/3,1/6,1/6,0)$ \\
33 & $(1/3,1/3,1/3,0)$ & $(1/3,1/3,1/3,0)$ & $(1/3,1/3,1/3,0)$ \\
34 & $(1/3,1/3,1/3,0)$ & $(1/3,1/3,1/3,0)$ & $(1/2,1/2,0,0)$ \\
35 & $(1/3,1/3,1/3,0)$ & $(1/3,1/3,1/3,0)$ & $(1,0,0,0)$ \\
36 & $(1/3,1/3,1/3,0)$ & $(2/5,1/5,1/5,1/5)$ & $(11/15,2/15,1/15,1/15)$ \\
37 & $(1/3,1/3,1/3,0)$ & $(5/12,1/4,1/6,1/6)$ & $(3/4,1/12,1/12,1/12)$ \\
38 & $(1/3,1/3,1/3,0)$ & $(5/12,5/12,1/12,1/12)$ & $(3/4,1/12,1/12,1/12)$ \\
39 & $(1/3,1/3,1/3,0)$ & $(4/9,1/3,1/9,1/9)$ & $(7/9,1/9,1/9,0)$ \\
40 & $(1/3,1/3,1/3,0)$ & $(1/2,1/6,1/6,1/6)$ & $(1/2,1/2,0,0)$ \\
41 & $(1/3,1/3,1/3,0)$ & $(1/2,1/6,1/6,1/6)$ & $(2/3,1/9,1/9,1/9)$ \\
42 & $(1/3,1/3,1/3,0)$ & $(1/2,1/6,1/6,1/6)$ & $(2/3,1/3,0,0)$ \\
43 & $(1/3,1/3,1/3,0)$ & $(1/2,1/2,0,0)$ & $(1/2,1/2,0,0)$ \\
44 & $(1/3,1/3,1/3,0)$ & $(1/2,1/2,0,0)$ & $(7/12,1/4,1/12,1/12)$ \\
45 & $(1/3,1/3,1/3,0)$ & $(1/2,1/2,0,0)$ & $(2/3,1/6,1/6,0)$ \\
46 & $(1/3,1/3,1/3,0)$ & $(5/9,2/9,1/9,1/9)$ & $(2/3,1/9,1/9,1/9)$ \\
47 & $(1/3,1/3,1/3,0)$ & $(2/3,1/3,0,0)$ & $(2/3,1/3,0,0)$ \\
48 & $(5/14,5/14,1/7,1/7)$ & $(3/7,2/7,2/7,0)$ & $(11/14,1/14,1/14,1/14)$ \\
49 & $(4/11,4/11,3/11,0)$ & $(5/11,2/11,2/11,2/11)$ & $(8/11,1/11,1/11,1/11)$ \\
50 & $(3/8,1/4,1/4,1/8)$ & $(1/2,1/2,0,0)$ & $(5/8,1/8,1/8,1/8)$ \\
51 & $(3/8,3/8,1/4,0)$ & $(5/8,1/8,1/8,1/8)$ & $(5/8,1/8,1/8,1/8)$ \\
52 & $(2/5,1/5,1/5,1/5)$ & $(2/5,2/5,1/5,0)$ & $(4/5,1/5,0,0)$ \\
53 & $(2/5,1/5,1/5,1/5)$ & $(1/2,1/2,0,0)$ & $(3/5,1/5,1/10,1/10)$ \\
54 & $(2/5,3/10,3/10,0)$ & $(2/5,2/5,1/10,1/10)$ & $(4/5,1/10,1/10,0)$ \\
55 & $(2/5,2/5,1/10,1/10)$ & $(3/5,1/5,1/5,0)$ & $(7/10,1/10,1/10,1/10)$ \\
56 & $(5/12,5/12,1/12,1/12)$ & $(1/2,1/4,1/4,0)$ & $(3/4,1/12,1/12,1/12)$ \\
57 & $(3/7,3/7,1/7,0)$ & $(4/7,1/7,1/7,1/7)$ & $(5/7,1/7,1/7,0)$ \\
58 & $(1/2,1/6,1/6,1/6)$ & $(1/2,1/2,0,0)$ & $(2/3,1/6,1/6,0)$ \\
59 & $(1/2,1/4,1/4,0)$ & $(1/2,1/2,0,0)$ & $(5/8,1/8,1/8,1/8)$ \\
60 & $(1/2,1/4,1/4,0)$ & $(1/2,1/2,0,0)$ & $(3/4,1/4,0,0)$ \\
61 & $(1/2,1/2,0,0)$ & $(1/2,1/2,0,0)$ & $(1/2,1/2,0,0)$ \\
62 & $(1/2,1/2,0,0)$ & $(1/2,1/2,0,0)$ & $(1,0,0,0)$ \\
63 & $(1/2,1/2,0,0)$ & $(5/8,1/8,1/8,1/8)$ & $(5/8,1/8,1/8,1/8)$ \\
64 & $(1/2,1/2,0,0)$ & $(2/3,1/6,1/6,0)$ & $(2/3,1/6,1/6,0)$ \\
65 & $(1,0,0,0)$ & $(1,0,0,0)$ & $(1,0,0,0)$ \\
    \bottomrule
  \end{tabular}
  \end{center}
\caption{Extreme rays $(V_A, V_B, V_C)$ of the moment cone for $\CC^4 \otimes \CC^4 \otimes \CC^4$ (with permutations removed).}
\label{tab:fourfourfour rays}
\end{table}

\begin{table}
\begin{center}
  \scriptsize
  \begin{tabular}{rlllrc}
    \toprule
    \multicolumn{1}{c}{$\#$} & \multicolumn{1}{c}{$H_A$} & \multicolumn{1}{c}{$H_B$} & \multicolumn{1}{c}{$H_C$} & \multicolumn{1}{c}{$z$} & \multicolumn{1}{c}{Remarks} \\
\midrule
1 & $(-5,-1,3,3)$ & $(-5,3,3,-1)$ & $(5,1,-3,-3)$ & 5 & $\star$ \\
2 & $(-5,-1,3,3)$ & $(1,-3,-3,5)$ & $(3,3,-1,-5)$ & 5 & - \\
3 & $(-5,3,-1,3)$ & $(-5,3,-1,3)$ & $(5,1,-3,-3)$ & 5 & $\star$ \\
4 & $(-5,3,-1,3)$ & $(-5,3,3,-1)$ & $(5,-3,1,-3)$ & 5 & $\star$ \\
5 & $(-5,3,-1,3)$ & $(-3,1,-3,5)$ & $(3,3,-1,-5)$ & 5 & $\star$ \\
6 & $(-5,3,-1,3)$ & $(-3,5,1,-3)$ & $(3,-5,3,-1)$ & 5 & $\star$ \\
7 & $(-5,3,-1,3)$ & $(1,-3,-3,5)$ & $(3,-1,3,-5)$ & 5 & - \\
8 & $(-5,3,-1,3)$ & $(1,-3,5,-3)$ & $(3,-1,-5,3)$ & 5 & - \\
9 & $(-5,3,3,-1)$ & $(-5,3,3,-1)$ & $(5,-3,-3,1)$ & 5 & $\star$ \\
10 & $(-5,3,3,-1)$ & $(-3,-3,1,5)$ & $(3,3,-1,-5)$ & 5 & $\star$ \\
11 & $(-5,3,3,-1)$ & $(-3,-3,5,1)$ & $(3,3,-5,-1)$ & 5 & $\star$ \\
12 & $(-5,3,3,-1)$ & $(-3,1,-3,5)$ & $(3,-1,3,-5)$ & 5 & $\star$ \\
13 & $(-5,3,3,-1)$ & $(-3,1,5,-3)$ & $(3,-1,-5,3)$ & 5 & $\star$ \\
14 & $(-5,3,3,-1)$ & $(-3,5,-3,1)$ & $(3,-5,3,-1)$ & 5 & $\star$ \\
15 & $(-5,3,3,-1)$ & $(-3,5,1,-3)$ & $(3,-5,-1,3)$ & 5 & $\star$ \\
16 & $(-5,3,3,-1)$ & $(-1,-5,3,3)$ & $(1,5,-3,-3)$ & 5 & $\star$ \\
17 & $(-5,3,3,-1)$ & $(-1,3,-5,3)$ & $(1,-3,5,-3)$ & 5 & $\star$ \\
18 & $(-5,3,3,-1)$ & $(-1,3,3,-5)$ & $(1,-3,-3,5)$ & 5 & $\star$ \\
19 & $(-3,-1,3,1)$ & $(-3,3,1,-1)$ & $(3,1,-1,-3)$ & 3 & $\star$ \\
20 & $(-3,-1,3,1)$ & $(1,-1,-3,3)$ & $(3,1,-1,-3)$ & 3 & - \\
21 & $(-3,1,1,1)$ & $(-3,1,1,1)$ & $(3,-1,-1,-1)$ & 3 & $\star$ \\
22 & $(-3,1,1,1)$ & $(-2,-2,2,2)$ & $(2,2,-2,-2)$ & 3 & $\star$ \\
23 & $(-3,1,1,1)$ & $(-2,2,-2,2)$ & $(2,-2,2,-2)$ & 3 & $\star$ \\
24 & $(-3,1,1,1)$ & $(-2,2,2,-2)$ & $(2,-2,-2,2)$ & 3 & $\star$ \\
25 & $(-3,1,1,1)$ & $(-1,-1,-1,3)$ & $(1,1,1,-3)$ & 3 & $\star$ \\
26 & $(-3,1,1,1)$ & $(-1,-1,3,-1)$ & $(1,1,-3,1)$ & 3 & $\star$ \\
27 & $(-3,1,1,1)$ & $(-1,3,-1,-1)$ & $(1,-3,1,1)$ & 3 & $\star$ \\
28 & $(-3,3,1,-1)$ & $(-3,3,1,-1)$ & $(3,-1,-3,1)$ & 3 & $\star$ \\
29 & $(-3,3,1,-1)$ & $(-1,-3,1,3)$ & $(3,1,-1,-3)$ & 3 & - \\
30 & $(-3,3,1,-1)$ & $(-1,-3,3,1)$ & $(1,3,-1,-3)$ & 3 & $\star$ \\
31 & $(-3,3,1,-1)$ & $(-1,-3,3,1)$ & $(3,1,-3,-1)$ & 3 & - \\
32 & $(-3,3,1,-1)$ & $(-1,3,1,-3)$ & $(1,-1,-3,3)$ & 3 & $\star$ \\
33 & $(-2,-2,2,2)$ & $(-2,2,2,-2)$ & $(1,1,-3,1)$ & 3 & $\star$ \\
34 & $(-2,2,-2,2)$ & $(-2,2,2,-2)$ & $(1,-3,1,1)$ & 3 & $\star$ \\
35 & $(-1,-1,-1,3)$ & $(0,0,0,0)$ & $(0,0,0,0)$ & 1 & $\star$ \\
36 & $(-1,0,0,1)$ & $(-1,1,0,0)$ & $(1,0,0,-1)$ & 1 & $\star$ \\
37 & $(-1,0,0,1)$ & $(0,0,-1,1)$ & $(1,0,0,-1)$ & 1 & - \\
38 & $(-1,0,1,0)$ & $(-1,0,1,0)$ & $(1,0,0,-1)$ & 1 & $\star$ \\
39 & $(-1,0,1,0)$ & $(-1,1,0,0)$ & $(1,0,-1,0)$ & 1 & $\star$ \\
40 & $(-1,0,1,0)$ & $(0,-1,0,1)$ & $(1,0,0,-1)$ & 1 & - \\
41 & $(-1,0,1,0)$ & $(0,-1,1,0)$ & $(0,1,0,-1)$ & 1 & $\star$ \\
42 & $(-1,0,1,0)$ & $(0,-1,1,0)$ & $(1,0,-1,0)$ & 1 & - \\
43 & $(-1,0,1,0)$ & $(0,0,-1,1)$ & $(0,1,0,-1)$ & 1 & $\star$ \\
44 & $(-1,1,0,0)$ & $(-1,1,0,0)$ & $(1,-1,0,0)$ & 1 & $\star$ \\
45 & $(-1,1,0,0)$ & $(0,-1,0,1)$ & $(0,1,0,-1)$ & 1 & $\star$ \\
46 & $(-1,1,0,0)$ & $(0,-1,1,0)$ & $(0,1,-1,0)$ & 1 & $\star$ \\
47 & $(-1,1,0,0)$ & $(0,0,-1,1)$ & $(0,0,1,-1)$ & 1 & $\star$ \\
48 & $(0,0,0,0)$ & $(0,0,0,0)$ & $(0,0,1,-1)$ & 0 & $\dagger$,$\star$ \\
49 & $(0,0,0,0)$ & $(0,0,0,0)$ & $(0,1,-1,0)$ & 0 & $\dagger$,$\star$ \\
50 & $(0,0,0,0)$ & $(0,0,0,0)$ & $(1,-1,0,0)$ & 0 & $\dagger$ \\
    \bottomrule
  \end{tabular}
  \end{center}
  \caption{Normal vectors $(H_A, H_B, H_C, z)$ of the facets of the moment cone for $\CC^4 \otimes \CC^4 \otimes \CC^4$ (with permutations removed). The last column states whether the facet includes the origin ($\dagger$) or the highest weight ($\star$).}
  \label{tab:fourfourfour facets}
\end{table}

We refer to \cite{onlinecode} for a complete list of our computational results.
In particular, we have verified the inequalities proposed previously by Klyachko for $a, b \leq 3$ \cite{Klyachko04}.
For $\calH = \CC^2 \otimes \CC^2 \otimes \CC^3 \otimes \CC^{12}$, which can be treated by a variant of the technique described above, we found that some of the proposed inequalities do \emph{not} correspond to extremal edges and are likely typographic mistakes (e.g, the second block on \cite[p.~43]{Klyachko04}).
This can also be deduced from the fact that they do not agree with Bravyi's inequalities when restricted to $\lambda_D = (1,0,\dots,0)$. 
We remark that the main challenge to obtaining concrete computational results in higher dimension is to find a tractable set of candidates beyond simple admissibility.

There are further variants of the one-body quantum marginal problem, such as for fermionic systems \cite{Coleman63,Ruskai69}, where the facets amount to strengthenings of the classical \emph{Pauli exclusion principle}. Our theory is also applicable to these scenarios, and it would be interesting to undertake a similar analysis of the facets that would allow the computation of the corresponding moment cones beyond what has been possible in the literature \cite{KlyachkoAltunbulak08}.
We refer to \cite[\S{}3]{Walter14} for initial investigations, where a family of fermionic inequalities has been proved for all local dimensions by using our method.
We remark that we have verified numerically that the pure-state fermionic inequalities list in \cite{KlyachkoAltunbulak08} are correct (but not their sufficiency).

\section{Horn Cone and Howe-Lee-Tan-Willenbring Invariants}
\label{sec:horn}

In this section, we consider the representation of $G = \GL(d) \times \GL(d) \times \GL(d)$ on $\calH = \mathfrak{gl}(d) \oplus \mathfrak{gl}(d)$ given by
\[ \Pi(g, h, k) (a,b) = (g a k^{-1}, h b k^{-1}), \] where $\mathfrak{gl}(d)$ is the Hilbert space of complex $d \times d$ matrices equipped with the trace inner product $\braket{a|b} := \tr a^\dagger b$.
We choose $K = \U(d) \times \U(d) \times \U(d)$ and take the maximal torus $T$ to consist of triples of unitary diagonal matrices, so that $i \mathfrak t$ can be identified with $\RR^d \oplus \RR^d \oplus \RR^d$.
We furthermore identify $i \mathfrak k \cong i \mathfrak k^*$ and $i \mathfrak t \cong i \mathfrak t^*$ by using the trace inner product and choose the usual positive roots $\alpha_{i,j}$, $i < j$, of each $\GL(d)$, such that the positive Weyl chamber $i \mathfrak t^*_+$ can be identified with triples of vectors with non-increasing entries each.
Finally, we set $\abs x :=\sum_i x_i$ and $x^* := (-x_d,\ldots, -x_1)$ for $x \in \RR^d$.
Using these conventions, the moment map for the $K$-action can be written as
\[ \mu_K(a, b) = (a a^\dagger, b b^\dagger, -a^\dagger a - b^\dagger b). \]
Any non-negative Hermitian matrix can be written in the form $a^\dagger a$; since the spectra of $a a^\dagger$ and $a^\dagger a$ are equal, the moment cone is equal to
\[ C(d) := \{ (\spec X, \spec Y, \spec Z) : X, Y \geq 0, Z \leq 0, ~ X + Y + Z = 0 \}, \]
where $\spec X$ denotes the eigenvalues of a Hermitian matrix $X$, ordered non-increasingly.
We will call $C(d)$ the \emph{Horn cone} in $d$ dimensions.
As proved in \cite{Klyachko98,KnutsonTao99}, cf.~\cite{KnutsonTao01,Belkale2006}, $C(d)$ is cut out by the Horn inequalities, which we will recall below.

We remark that the Horn cone as defined thus is not maximal-dimensional but rather supported in the linear subspace $\{ (x,y,z) : \abs x + \abs y + \abs z = 0 \}$.
Dually, any two normal vectors that differ by $(\id, \id, \id)$ determine the same facet.
This can be avoided by working with the subgroup $G' = \GL(d) \times \GL(d) \times \SL(d) \subseteq G$, but we will not discuss this point any further.

For any subset $I \subseteq [d] := \{1,\dots,d\}$, we define $\CC^I := \bigoplus_{i \in I} \CC e_i \subseteq \CC^d$.
We define $E_I$ to be the orthogonal projection onto $\CC^I$ 
and abbreviate $E_r := E_{[r]}$ for any $r \leq d$.
It is easy to see (e.g., by using an argument involving restricted roots as in \cref{subsec:qmp candidates}) that apart from positivity ($x_d \geq 0$, etc.), any non-trivial facet of the Horn cone is of the form $(H_{IJK}, -) \geq 0$, where $H_{IJK} := (E_I, E_J, E_K)$ for subsets $I, J, K \subseteq [d]$ with $r := \abs I = \abs J = \abs K < d$. 


\subsection{Trace and Horn Condition}
\label{subsec:trace horn horn}

In the following, we will show that the trace condition and the Horn condition assume their familiar form in the context of Horn's problem (e.g., \cite{KnutsonTao01}), thereby justifying our terminology.
For any $I = \{ i_1 < \dots < i_r \} \subseteq [d]$ of cardinality $r$, let $w_I$ denote the permutation that sends $[r]$ to $I$ and $[r]^c$ to the complement $I^c = \{ i^c_1 < \dots < i^c_{d-r} \}$ while preserving the order of each block.
Then $E_I = w_I \cdot E_r = w_I E_r w_I^{-1}$ (we identify $w_I$ with a permutation matrix).
However, $w_I$ does \emph{not} satisfy the conditions of \cref{cor:trace criterion}.
Instead, we shall write $E_I = \widetilde w_I \cdot E_r$ with $\widetilde w_I := w_{I^c} w_0$, which reverses the order of each block.
Then we find that $H_{IJK}$ satisfies the trace condition if and only if
\[ \ell(w_{I^c}) + \ell(w_{J^c}) + \ell(w_{K^c}) = \dim \calH(H_r < 0), \]
where $H_r := (E_r, E_r, E_r)$.
For the left-hand side, we compute
\begin{align*}
  \ell(w_{I^c})
  &= \#\{ (b,a) : b < a, w_{I^c}(b) > w_{I^c}(a) \}
  = \sum_{a=1}^r \#\{ j \in I^c : i_a < j \} \\
  &= \sum_{a=1}^r \left( (d - r) - (i_a - a) \right)
  = \abs{\lambda_I},
\end{align*}
where we have defined $\lambda_I := (d - r - (i_a - a))_{a=1}^r$, which is a sequence of non-increasing non-negative integers.
For the right-hand side, observe that $\calH(H_r < 0)$ consists of pairs of block matrices of the form $\left( \begin{smallmatrix} 0 & 0 \\ * & 0\end{smallmatrix} \right)$, hence is isomorphic to $M_{d-r,r}^2$, the vector space of pairs of complex $(d-r) \times r$-matrix.
We conclude that the trace condition for $H_{IJK}$ amounts to
\begin{equation}
\label{eq:trace horn}
  \abs{\lambda_I} + \abs{\lambda_J} + \abs{\lambda_K} = 2 (d-r) r.
\end{equation}
It is not hard to see that \eqref{eq:trace horn} is precisely equivalent to Horn's trace condition.


To analyze the Horn condition, we note that the centralizer of $H_{IJK}$ is
$K(H_{IJK}=0) = \left( \U(I) \times \U(I^c) \right) \times \left( \U(J) \times \U(J^c) \right) \times \left( \U(K) \times \U(K^c) \right)$
where we denote by $U(I)$ the subgroup of unitaries that act non-trivially only on $\CC^I \subseteq \CC^d$, etc.
It will be convenient to use the three isomorphisms
\begin{equation}
\label{eq:horn isos}
\begin{aligned}
  &\U(r)^3 \times \U(d-r)^3 \rightarrow K(H_{IJK} = 0), \quad (U,V,W,U',V',W') \mapsto \\
  &\qquad\qquad\qquad (w_I \left( \begin{smallmatrix} U & 0 \\ 0 & U' \end{smallmatrix} \right) w_I^{-1},
   w_J \left( \begin{smallmatrix} V & 0 \\ 0 & V' \end{smallmatrix} \right) w_J^{-1},
   w_K \left( \begin{smallmatrix} W & 0 \\ 0 & W' \end{smallmatrix} \right) w_K^{-1}),\\
  &M_{d-r,r}^2 \rightarrow \calH(H_{IJK} < 0), \quad
  (A, B) \mapsto
  (w_I \left( \begin{smallmatrix} 0 & 0 \\ A & 0 \end{smallmatrix} \right) w_K^{-1},
   w_J \left( \begin{smallmatrix} 0 & 0 \\ B & 0 \end{smallmatrix} \right) w_K^{-1}),\\
  &T_I \times T_J \times T_K \rightarrow \mathfrak n_-(H_{IJK} < 0), \quad (X, Y, Z) \mapsto \\
  &\qquad\qquad\qquad(w_I \left( \begin{smallmatrix} 0 & 0 \\ X & 0 \end{smallmatrix} \right) w_I^{-1},
   w_J \left( \begin{smallmatrix} 0 & 0 \\ Y & 0 \end{smallmatrix} \right) w_J^{-1},
   w_K \left( \begin{smallmatrix} 0 & 0 \\ Z & 0 \end{smallmatrix} \right) w_K^{-1}),
\end{aligned}
\end{equation}
where we have defined
$T_I := \{ X \in M_{d-r,r} : \braket{b | X | a} = 0 \text{ if } i_a > i^c_b \}$
etc.
A short computation then reveals that the element $\kappa_{H_{IJK}}$ of \cref{prp:horn condition} identifies with the highest weight
\begin{equation}
\label{eq:kappa horn}
  \kappa_{IJK} := (\lambda_I, \lambda_J, \lambda_K - 2 (d-r) \chi_r) \oplus (\lambda_{I^c}, \lambda_{J^c}, \lambda_{K^c} - r \chi_{d-r})
\end{equation}
of $\U(r)^3 \times \U(d-r)^3$, where $\chi_r$ denotes the weight of the determinant representation of $\U(r)$.
%
On the other hand, by using the isomorphism
\begin{equation}
\label{eq:horn isos two}
\begin{cases}
  \gl(r)^2 \oplus \gl(d-r)^2 \rightarrow \calH(H_{IJK} = 0) \\
   (a, b, a', b') \mapsto
   (w_I \left( \begin{smallmatrix} a & 0 \\ 0 & a' \end{smallmatrix} \right) w_K^{-1},
    w_J \left( \begin{smallmatrix} b & 0 \\ 0 & b' \end{smallmatrix} \right) w_K^{-1})
\end{cases}
\end{equation}
the moment cone $K(H_{IJK}=0)$ on $\calH(H_{IJK}=0)$ gets likewise identified with the direct product of the Horn cones $C(r) \times C(d-r)$.
Thus the Horn condition (\cref{prp:horn condition}) asserts that
$(\lambda_I, \lambda_J, \lambda_K - 2 (d-r) \chi_r) \in C(r)$
as well as
$(\lambda_{I^c}, \lambda_{J^c}, \lambda_{K^c} - r \chi_{d-r}) \in C(d-r)$.
These two conditions are in fact equivalent by a well-known duality of the Littlewood-Richardson coefficients. 
Therefore, we arrive at a single condition
\begin{equation}
\label{eq:horn horn}
  (\lambda_I, \lambda_J, \lambda_K - 2 (d-r) \chi_r) \in C(r).
\end{equation}
This condition is not only necessary for $H_{IJK}$ to be a facet, but it is also sufficient for $H_{IJK}$ to be a valid inequality, as is well-known (e.g., \cite{KnutsonTao01}). We will show in the next section how in the context of our work this can be deduced from the saturation conjecture and Schubert calculus (in fact, we shall see that \eqref{eq:horn horn} implies that $H_{IJK}$ is a Ressayre element).
In particular, we obtain the familiar recursive definition of Horn's inequalities:
Define $\operatorname{Horn}(d, r)$ to be the set of all triples $I, J, K \subseteq [d]$ of cardinality $r < d$ that satisfy the trace condition \eqref{eq:trace horn} as well as
\[ \sum_{a \in A} i_a + \sum_{b \in B} j_b + \sum_{c \in C} k_c \leq s(d+1) + \frac {s (s+1)} 2 \]
for all $s < r$ and all triples $(A,B,C) \in \operatorname{Horn}(r, s)$.
Then \eqref{eq:horn horn} implies that $H_{IJK}$ satisfies the Horn condition if and only if $(I,J,K) \in \operatorname{Horn}(d, r)$.
Therefore, our trace and Horn conditions are indeed a generalization of the classical conditions due to Horn.

\subsection{The Determinant Polynomial}

In this section, we will show that any $H_{IJK}$ that satisfies the trace and Horn condition is automatically a Ressayre element, i.e., that the determinant polynomial $\delta_{IJK} := \delta_{H_{IJK}}$ is non-zero.
To start, we note that by the saturation property of the Littlewood-Richardson coefficients \cite{KnutsonTao99} (cf.~\cite{Belkale2006}), the Horn condition \eqref{eq:horn horn} implies that the following space of $\GL(r)$-invariants is non-zero:
\begin{equation}
\label{eq:horn invariants nonzero}
  \big( V_{\lambda_I} \otimes V_{\lambda_J} \otimes V_{\lambda_K - 2 (d-r) \chi_r} \big)^{\GL(r)} \neq 0
\end{equation}
A natural candidate is certainly the determinant polynomial itself (from which we had obtained the Horn condition), but it is not obvious that the Horn condition should imply that $\delta_{IJK} \neq 0$.
We will give two alternative arguments that show that this is indeed the case.

The first proof follows an argument of Belkale \cite{Belkale2006}. We start with the observation that
\begin{align*}
  \dim \big( V_{\lambda_I} \otimes V_{\lambda_J} \otimes V_{\lambda_K - 2 (d-r) \chi_r} \big)^{\GL(r)}
  &= \dim \big( V_{\lambda_I} \otimes V_{\lambda_J} \otimes V_{\lambda_K} \big)^{\SL(r)} \\
  &= \dim \big( V_{\lambda_{I_-}} \otimes V_{\lambda_{J_-}} \otimes V_{\lambda_{K_-}} \big)^{\SL(r)}
\end{align*}
where $I_- := \{ d + 1 - i : i \in I \}$, etc.; the second equality follows from $\lambda_{I_-} = (d-r) \chi_r + \lambda_I^*$.
Now consider the Grassmannian $\Gr(r, d)$, which consists of the $r$-dimensional subspaces of $\CC^d$.
It is a homogeneous $\GL(d)$-manifold of real dimension $2 r (d-r)$.
Let $B$ denote the standard (upper-triangular) Borel and denote by $V_I$ the point in $\Gr(r, d)$ corresponding to the subspace $\CC^I \subseteq \CC^d$; set $V_r := V_{[r]}$.
For any $w \in S_d$, we define Schubert cell $\Omega^0(w) := B w \cdot V_r$ and denote by $[\Omega(w)]$ the cohomology class determined by its closure.
We now use the classical result that the Littlewood-Richardson coefficients can be computed as an intersection number of Schubert cells \cite[\S{}9.4]{Fulton97}:
\[ \dim \big( V_{\lambda_{I_-}} \otimes V_{\lambda_{J_-}} \otimes V_{\lambda_{K_-}} \big)^{\SL(r)} = [\Omega(w_{I_-})] \cap [\Omega(w_{J_-})] \cap [\Omega(w_{K_-})]. \]
Using that $\Omega(w_{I_-}) = \Omega(w_0 w_I)$, we recognize that \eqref{eq:horn invariants nonzero} is equivalent to the cohomological condition
\begin{equation}
\label{eq:horn cohomological consequence}
  [\Omega(w_0 w_I)] \cap [\Omega(w_0 w_J)] \cap [\Omega(w_0 w_K)] \neq 0.
\end{equation}
Now consider the ``action map''
\begin{equation}
\label{eq:abstract action map}
N_-(H_{IJK} < 0) \times \calH(H_{IJK} \geq 0) \rightarrow \calH, (g,\phi) \mapsto \Pi(g) \phi,
\end{equation}
a close variant of what we had analyzed in the proof of \cref{prp:partial converse}.
We will use the parametrizations
\begin{align*}
T_I \times T_J \times T_K \rightarrow N_-(H_{IJK} < 0), \quad
&(x,y,z) \mapsto (
w_I \widetilde x w_I^{-1},
w_J \widetilde y w_J^{-1},
w_K \widetilde z w_K^{-1}) \\
(M_{d-r,r}^\perp)^2 \rightarrow \calH(H_{IJK} \geq 0), \quad
&(p, q) \mapsto (w_I p w_K^{-1}, w_J q w_K^{-1}),
\end{align*}
where $\widetilde x := \left( \begin{smallmatrix} \id & 0 \\ x & \id \end{smallmatrix} \right)$ etc.\ and $M_{d-r,r}^\perp := \{ \left( \begin{smallmatrix} * & * \\ 0 & * \end{smallmatrix} \right) \}$. In explicit coordinates $p = \left(\begin{smallmatrix}a & a''\\0 & a'\end{smallmatrix}\right)$ and $q = \left(\begin{smallmatrix}b & b'' \\ 0 & b'\end{smallmatrix}\right)$, the map \eqref{eq:abstract action map} becomes
\begin{equation}
  \label{eq:horn action map}
  R_{IJK} \colon \begin{cases}
  T_I \times T_J \times T_K \times (M_{d-r,r}^\perp)^2 \rightarrow \gl(d)^2 \\
  (x, y, z, p, q) \mapsto (\widetilde x p \widetilde z^{-1}, \widetilde y q \widetilde z^{-1}) =
  (
  \big( \begin{smallmatrix}
    a - a'' z & a'' \\
    x a - a' z - x a'' z \;&\; x a'' + a'
  \end{smallmatrix} \big), \\
  \qquad\qquad\qquad\qquad\qquad\qquad\qquad\quad\big( \begin{smallmatrix}
    b - b'' z & b'' \\
    y b - b' z - y b'' z \;&\; y b'' + b'
  \end{smallmatrix} \big)
  )
  \end{cases}
\end{equation}

\begin{lemma}
\label{lem:dominant}
  The map $R_{IJK}$ is dominant.
\end{lemma}
\begin{proof}
  We show that the fibers of $R_{IJK}$ are generically non-empty.
  Thus let $(g, h) \in \GL(d)^2 \subseteq \gl(d)^2$.
  Then $(x, y, z, p, q)$ is an element of the fiber $R_{IJK}^{-1}(g, h)$ if and only if
  \[ g^{-1} \widetilde x p = h^{-1} \widetilde y q = \widetilde z. \]
  Any such $p$ and $q$ is automatically invertible, and therefore a general element in the stabilizer group of $V_r \in \Gr(r, d)$.
  Thus the fiber is non-empty if and only if we can find $(x,y,z)$ such that
  \[ g^{-1} \widetilde x \cdot V_r = h^{-1} \widetilde y \cdot V_r = \widetilde z \cdot V_r. \]
  Since $\{ \widetilde x \cdot V_r : x \in T_I \} = (w_0 w_I)^{-1} \Omega(w_0 w_I)$ etc., this is the case if and only if 
  \[ g^{-1} (w_0 w_I)^{-1} \Omega(w_0 w_I) \cap h^{-1} (w_0 w_J)^{-1} \Omega(w_0 w_J) \cap (w_0 w_I)^{-1} \Omega(w_0 w_I) \neq \emptyset. \]
  By Kleiman's transversality theorem, the cohomological condition \eqref{eq:horn cohomological consequence} ensures that this is the case for generic $(g,h)$.
\end{proof}

Now observe that using the identifications \eqref{eq:horn isos} and \eqref{eq:horn isos two}, the tangent map \eqref{eq:tangent map} at some base point $(a,b,a',b') \in \gl(r)^2 \oplus \gl(d-r)^2$ reads
\begin{equation}
\label{eq:horn tangent map}
  V_{IJK}(a,b,a',b') \colon \begin{cases}
    T_I \oplus T_J \oplus T_K \rightarrow M_{d-r,r}^2 \\
    (X, Y, Z) \mapsto (X a - a' Z, Y b - b' Z)
  \end{cases}
\end{equation}

\begin{corollary}
  The determinant polynomial $\delta_{IJK} = \det V_{IJK}$ is non-zero, i.e., $H_{IJK}$ is a Ressayre element.
\end{corollary}
\begin{proof}
  By Sard's theorem and \cref{lem:dominant}, there exists a point where the differential of $R_{IJK}$ is surjective.
  By writing the differential of \eqref{eq:horn action map} in coordinates and comparing with \eqref{eq:horn tangent map}, it is not hard to see that surjectivity of the former at some point $(x,y,z,p,q)$ implies surjectivity of the tangent map at $(a,b,a',b')$, where $p = \left( \begin{smallmatrix} a & a'' \\ 0 & a' \end{smallmatrix} \right)$ and $q = \left( \begin{smallmatrix} b & b'' \\ 0 & b' \end{smallmatrix} \right).$
  Thus $H_{IJK}$ is a Ressayre element.
\end{proof}

\begin{example}
\label{ex:invariant}
Let $d=6$, $r=3$, and $I = J = K = \{ 1 < 3 < 5 \}$, so that $I^c = J^c = K^c = \{ 2 < 4 < 6 \}$.
Then $\lambda_I = \lambda_J = \lambda_K = (3,2,1)$ and $\lambda_{I^c} = \lambda_{J^c} = \lambda_{K^c} = (2,1,0)$, and the associated Littlewood-Richardson coefficients are given by
\begin{align*}
&\quad\dim \big( V_{(3,2,1)} \otimes V_{(3,2,1)} \otimes V_{(-3,-4,-5)} \big)^{\GL(3)} \\
&= \dim \big( V_{(2,1,0)} \otimes V_{(2,1,0)} \otimes V_{(-1,-2,-3)} \big)^{\GL(3)} = 2.
\end{align*}
In \cref{tab:covariants} we list a set of homogeneous generators of the algebra of lowest weight vectors in $R(\gl(3)^2)$.
Note that $F_1(x, y) := f_1(x) f_2(y) h(y, x)$ and $F_2(y, x) := F_1(x, y)$ span the two-dimensional subspace of weight $((-2,-1,0)$, $(-2,-1,0)$, $(1,2,3))$.
An explicit calculation best left to a computer algebra system \cite{sage} verifies that
\[ \delta_{IJK}(a,b,a',b') = -f_3(a) f_3(b) \left( F_1(a, b) F_2(a', b') - F_2(a, b) F_1(a', b') \right). \]
Therefore, in agreement with \eqref{eq:kappa horn} $\delta_{IJK}$ is indeed a lowest weight vector of weight
\begin{align*}
  -\kappa_{IJK} =~&((-3,-2,-1), (-3,-2,-1), (3,4,5)) \\
  \oplus~&((-2,-1,0),(-2,-1,0),(1,2,3)).
\end{align*}
Note that if we consider $a'$ and $b'$ as coefficients rather than indeterminates, $\delta_{IJK}(-,a',b')$ spans the two-dimensional subspace of lowest weight vectors of weight $((-3,-2,-1),(-3,-2,-1),(3,4,5))$ as we vary $a'$ and $b'$ over $\mathfrak{gl}(3)$.
Likewise, $\delta_{IJK}(a,b,-)$ spans the subspace of lowest weight vectors of weight $((-2,-1,0),(-2,-1,0),(1,2,3))$ if we instead vary $a$ and $b$.
\end{example}

\begin{table}
\begin{center}
\begin{tabular}{ll}
  \toprule
  \multicolumn{1}{c}{Generators} & \multicolumn{1}{c}{Weight} \\
  \midrule
  $f_1(x) = x_{1,3}$ & $((-1,0,0), (0,0,0), (0,0,1))$ \\
  $f_2(x) = \det \left(\begin{smallmatrix}x_{1,2} & x_{2,3} \\ x_{2,2} & x_{2,3}\end{smallmatrix}\right)$ & $((-1,-1,0), (0,0,0), (0,1,1))$ \\
  $f_3(x) = \det x$ & $((-1,-1,-1), (0,0,0), (1,1,1))$ \\
  $g(x,y) = x_{1,2} y_{1,3} - y_{1,2} x_{1,3}$ & $((-1,0,0), (-1,0,0), (0,1,1))$ \\
  $h(x,y) = x_{1,1} \det \left(\begin{smallmatrix} y_{1,2} & y_{1,3} \\ y_{2,2} & y_{2,3} \end{smallmatrix}\right)$ & $((-1,0,0), (-1,-1,0), (1,1,1))$ \\
          $\quad- x_{1,2} \det \left(\begin{smallmatrix} y_{1,1} & y_{1,3} \\ y_{2,1} & y_{2,3} \end{smallmatrix}\right)
          + x_{1,3} \det \left(\begin{smallmatrix} y_{1,1} & y_{1,2} \\ y_{2,1} & y_{2,2} \end{smallmatrix}\right)$ & $((-1,0,0), (-1,-1,0), (1,1,1))$ \\
\bottomrule
\end{tabular}
\end{center}
\caption{A set of homogeneous generators of the algebra of $N_-(3)^3$-invariant polynomials on $\mathfrak{gl}(3)^2$ together with their weights (up to permutation $x \leftrightarrow y$).}
\label{tab:covariants}
\end{table}

The phenomenon just observed in \cref{ex:invariant} is in fact of a general nature.
In \cite{HoweTanWillenbring2005,HoweLee2007} for any triple of Young diagrams $(D,E,F)$ a matrix determinant $\Delta_{(D,E,F),(A,B)}(X,Y)$  had been constructed that likewise depends on two pairs of matrices (cf.~\cite{HoweLee2012,Lee13}).
For each choice of $A$ and $B$, this determinant is a highest weight vector of weight $(D^T, E^T, F^T)$, where the superscript denotes the transpose diagrams.
Moreover, as one varies $A$ and $B$, the $\Delta_{(D,E,F),(A,B)}$ span the corresponding subspace of highest weight vectors in the \emph{tensor product algebra} of $\GL(r)$, whose dimension is equal to the corresponding Littlewood-Richardson coefficient
$\dim ( V_{D^T} \otimes V_{E^T} \otimes V_{F^T} )^{\GL(r)}$.
In fact, their construction is related by a trivial transform to our determinant polynomial:
\begin{equation}
\label{eq:howe mapping}
\begin{aligned}
  \Delta_{(D,E,F),(A,B)}(X, Y) =~&\delta_{IJK}(w_0 X^{-1}, w_0 Y^{-1}, -A^T w_0, -B^T w_0) \\
  &(\det XY)^{d-r}
\end{aligned}
\end{equation}
where $(D,E,F) = (\lambda_{I^c}, \lambda_{J^c}, \lambda_{K^c}^* + r \chi_{d-r})$.
Equation \eqref{eq:howe mapping} can be shown by manual inspection, relating the matrix elements of the tangent map \eqref{eq:tangent map} with the matrix constructed by Howe et al.
This gives a geometric interpretation of the invariants constructed in \cite{HoweTanWillenbring2005} -- namely, as the determinant of the tangent map \eqref{eq:horn tangent map} associated with a Ressayre element --, and it also serves as an alternative, second proof that $\delta_{IJK} \neq 0$ is implied by the Horn condition via \eqref{eq:horn invariants nonzero}.

\bigskip

\noindent\textbf{Acknowledgments.} We would like to thank Velleda Baldoni, Soo Teck Lee, Nicolas Ressayre, and Jonathan Skowera for pleasant discussions.
We acknowledge the National Technological University and the Institute for Mathematical Sciences at the National University of Singapore for their hospitality during the program on Inverse Moment Problems, where this work had been initiated.
MW acknowledges financial support by the Swiss National Science Foundation (grants PP00P2-128455, 20CH21-138799 (CHIST-ERA project CQC)),
the Swiss National Center of Competence in Research `Quantum Science and Technology (QSIT)',
the Swiss State Secretariat for Education and Research supporting COST action MP1006,
the European Research Council under the European Union's Seventh Framework Programme (FP/2007--2013)/ERC Grant Agreement no.~337603,
the Simons Foundation,
and FQXi.

\appendix
\section{On the Dimension of the Kronecker Cone}
\label{appendix}

\begin{lemma}
  Let $1 < a \leq b \leq c \leq ab$. Then there exists an operator $M \colon \CC^a \otimes \CC^b \rightarrow \CC^c$ such that the following is true:
  For any $U_A \in \U(a)$, $U_B \in \U(b)$, and $U_C \in \U(c)$,
  \begin{equation}
  \label{eq:stabilization}
     U_C M (U_A \otimes U_B) = M
  \end{equation}
  implies that $U_A$, $U_B$, and $U_C$ are scalars (i.e., proportional to the identity matrix).
\end{lemma}
\begin{proof}
  Let $e_i$ denote the standard basis vectors in any $\CC^d$. We define the following $b$ orthonormal vectors in $\CC^a \otimes \CC^b$:
  \begin{align*}
    v_1 &:= e_1 \otimes e_1, \dots, v_{a-1} := e_{a-1} \otimes e_{a-1}, \\
    v_a &:= e_{a-1} \otimes e_a, \dots, v_{b-1} := e_{a-1} \otimes e_{b-1}, \\
    v_b &:= \sum_{\text{all other}} e_i \otimes e_j.
  \end{align*}
  Set $V := \Span_\CC \{ v_j \}$.
  Define an operator $M$ with $M v_j = \alpha_j e_j$ for some non-zero coefficients $\alpha_j$ to be determined later and which sends $V^\perp$ surjectively to $\Span \{ e_{b+1}, \dots, e_c \}$.
  This is always possible since $\dim V^\perp = ab - b \geq c - b$, and the resulting operator $M$ is surjective.
  Note that $M^\dagger M$ is the direct sum of two Hermitian submatrices; the first acts on $V$, sending $v_j \mapsto \abs{\alpha_j}^2 v_j$, the second acts on $V^\perp$ and is independent of our choice of $\alpha_j$. Thus we may choose the $\alpha_j$ such that their absolute values squared are pairwise distinct and also distinct from the eigenvalues of the second submatrix.

  Now suppose that $U_A$, $U_B$, and $U_C$ are unitaries such that \eqref{eq:stabilization} holds. Then,
  \[ M^\dagger M = (U_A \otimes U_B)^\dagger M^\dagger M (U_A \otimes U_B), \]
  i.e., $M^\dagger M$ commutes with $U_A \otimes U_B$.
  But then our choice of $\alpha_j$ implies that $U_A \otimes U_B$ leaves each of the eigenspaces $\CC v_1, \dots, \CC v_b$ stable.
  From the first $b-1$ eigenvectors, which are tensor products, it follows that $U_A e_1 \in \CC e_1, \dots, U_A e_{a-1} \in \CC e_{a-1}$ as well as $U_B e_1 \in \CC e_1, \dots, U_B e_{b-1} \in \CC e_{b-1}$ (this requires $a > 1$).
  This implies that $U_A$ and $U_B$ are diagonal, since they are unitaries.
  Thus the fact that $v_b$ is also an eigenvector shows that $U_A$ and $U_B$ in fact act by scalars.
  Finally, observe that this in turn implies that $U_C M \in \CC M$.
  Since $M$ is surjective, we conclude that also $U_C$ acts by a scalar.
\end{proof}

\begin{corollary}
\label{cor:maxdim}
  Let $1 < a \leq b \leq c \leq ab$. Then the Kronecker cone $C(a, b, c)$ is maximal-dimensional.
\end{corollary}

In the case where $a=b=c$ the following lemma strengthens \cref{cor:maxdim}:

\begin{lemma}
\label{lem:origin ball}
  Let $\tau_d = \id/d \in \RR^d$.
  The only facets of $C(d,d,d)$ that contain the point $(\tau_d,\tau_d,\tau_d)$ are the trivial facets.
\end{lemma}
\begin{proof}
  For any $\lambda_A \in i \mathfrak t^*$ with $\lambda_d \geq 0$ and $\abs\lambda = 1$, there exists a unit vector $\psi \in \CC^d \otimes \CC^d \otimes \CC^d$ whose one-body reduced density matrices have spectra $(\lambda_A,\tau_d,\tau_d)$ \cite{rectangular11}.
  It follows that
  $(\tau_d + \varepsilon_A, \tau_d, \tau_d) \in C(d,d,d)$
  for any perturbation $\varepsilon_A \in i \mathfrak t(d)^*_+$ whose components are small enough in absolutely value.
  By permutation symmetry and convexity, we obtain that
  \[ (\tau_d + \varepsilon_A, \tau_d + \varepsilon_B, \tau_d + \varepsilon_C) \in C(d,d,d) \]
  for any triple of small perturbations $\varepsilon_A, \varepsilon_B, \varepsilon_C$.
  Therefore, the only constraints in the vicinity of $(\tau_d,\tau_d,\tau_d)$ are the Weyl chamber inequalities.
\end{proof}

\bibliographystyle{unsrtnat}
\bibliography{momentcones}

\address{Mich\'ele Vergne: Universit\'e Paris 7 Diderot, Institut Math\'{e}matique de Jussieu, Sophie Germain, Case 75205, Paris Cedex 13, France \\
\email{michele.vergne@imj-prg.fr}}

\address{Michael Walter: Institute for Theoretical Physics, Stanford University, Stanford, CA 94305, USA;
Institute for Theoretical Physics, ETH Zurich, Wolfgang-Pauli-Str.~27, 8093 Z\"urich, Switzerland\\
\email{michael.walter@stanford.edu}}

\end{document}